\documentclass[12pt,letterpaper,reqno,english]{amsart}
\usepackage[T1]{fontenc}
\usepackage{babel}
\usepackage{textcomp}
\usepackage{amstext}
\usepackage{amsthm}
\usepackage{esint}
\usepackage[unicode=true,pdfusetitle,
 bookmarks=true,bookmarksnumbered=false,bookmarksopen=false,
 breaklinks=false,pdfborder={0 0 1},backref=false,colorlinks=false]
 {hyperref}

\makeatletter

\pdfpageheight\paperheight
\pdfpagewidth\paperwidth

\newcommand{\lyxmathsym}[1]{\ifmmode\begingroup\def\b@ld{bold}
  \text{\ifx\math@version\b@ld\bfseries\fi#1}\endgroup\else#1\fi}

\theoremstyle{definition}
 \newtheorem{example}{\protect\examplename}
\theoremstyle{plain}
\newtheorem{thm}{\protect\theoremname}
\theoremstyle{plain}
\newtheorem{cor}{\protect\corollaryname}
\theoremstyle{plain}
\newtheorem{lem}{\protect\lemmaname}

\usepackage{tikz}
\usepackage{textgreek, enumitem}
\usepackage[T1]{fontenc}
\usetikzlibrary{decorations.pathreplacing, calc,fit,shapes, positioning,arrows, patterns}

\@ifundefined{definecolor}
 {\usepackage{color}}{}
\usepackage{amsfonts}
\usepackage{ifpdf} 
\ifpdf 
 \IfFileExists{lmodern.sty}
  {\usepackage{lmodern}}{}
\fi 

\providecommand{\U}[1]{\protect \rule{.1in}{.1in}}

\usepackage{bbold}

\DeclareMathOperator{\E}{\mathbb{E}}
\DeclareMathOperator{\Prob}{\mathbb{P}}
\DeclareMathOperator{\R}{\mathbb{R}}

\newcommand{\euler}{\mathrm{e}}

\usepackage{setspace}

\makeatother

\providecommand{\corollaryname}{Corollary}
\providecommand{\examplename}{Example}
\providecommand{\lemmaname}{Lemma}
\providecommand{\theoremname}{Theorem}

\begin{document}
\title{Fuzzy Conventions}
\author{Marcin Pęski}
\date{\today}
\begin{abstract}
We study binary coordination games with random utility played in networks.
A typical equilibrium is fuzzy - it has positive fractions of agents
playing each action. The set of average behaviors that may arise in
an equilibrium typically depends on the network. The largest set (in
the set inclusion sense) is achieved by a network that consists of
a large number of copies of a large complete graph. The smallest set
(in the set inclusion sense) is achieved on a lattice-type network.
It consists of a single outcome that corresponds to a novel version
of risk dominance that is appropriate for games with random utility.
\thanks{TBA}
\end{abstract}

\maketitle
\smallskip{}

\section{Introduction}

An individual's behavior in social or economic situations is often
positively influenced by similar decisions made by their friends,
acquaintances, or neighbors. Examples include the decision to maintain
a neat front yard, to obey speed limits or tax laws, or to engage
in criminal activity. A substantial literature has shown that the
details of the network of social interactions may affect which of
the equilibria is more likely to arise (see, for example, references
in \cite{jackson_games_2015}). A typical result in this literature
establishes conditions under which a particular behavior is adopted
by everybody and becomes a convention (see \cite{young_evolution_1993},
\cite{ellison_learning_1993}, among many others). At the same time,
a completely uniform behavior is very rare in the real world. Even
in situations which clearly involve positive externalities, there
will often be interactions in which neighbors make the opposite choices. 

An obvious reason for heterogeneous behavior is that individuals are
different and their tastes and unique circumstances play just as important
of a role in determining their decisions as the behavior of their
neighbors. The goal of this paper is to analyze the impact of heterogeneity
in a systematic way. A natural question is how adding heterogeneity
in tastes affects our ability to predict the unique outcome. What
can we say about the set of possible equilibrium conventions and how
does it depend on the network, and other parameters of the model,
like taste distribution?

To address these questions, we study a random utility coordination
game played in a network. Each player chooses a binary action and
the relative gain from the action is increasing in the fraction of
neighbors who make the same choice. Additionally, as in the literature
on random choice, payoffs are subject to individual i.i.d. shocks.
The independence assumption is key for our results and it is appropriate
for some, but not all applications. An individual's equilibrium action
as well as the aggregate distribution of equilibrium actions depend
on the realization of the entire profile of payoff shocks. We are
interested in the asymptotics of the average (i.e., aggregate) behavior
as the network becomes arbitrarily large and, importantly, the graph
becomes sufficiently fine, i.e., the weight of the largest neighbor
in an neighborhood of each player becomes sufficiently small. The
latter ensures that no single individual has a disproportionate impact
on another and it is the second key assumption in our model. 

In contrast to simple model of coordination games, a typical equilibrium
in our model is fuzzy - it has positive fractions of populations playing
each action. Also, despite there being only two potential actions,
a coordination game may have many more than two equilibria. To illustrate
the latter point, consider a continuum toy version of the model, in
which individual payoffs depend on the fraction $x$ of agents choosing
the high action in the entire population. Let $P\left(x\right)$ be
the probability of a payoff shock for which the agent best response
is to choose the high action as well. Function $P$ has values between
0 and 1 and is increasing in $x$, but is otherwise arbitrary. An
example is illustrated on Figure \ref{fig:Continuum-best-response}.
Fixed points of $P$, i.e., intersections of the graph of $P$ with
$45^{\circ}$ diagonal, correspond to equilibria of the toy model.
\begin{figure}
\begin{centering}
\begin{tikzpicture} [scale = 0.6, dot/.style={circle,inner sep=1.3pt,fill},   extended line/.style={shorten >=-10,shorten <=-10}] 

	\draw [->, thick] (0,0)  -- (11,0) node [below right, scale=1] {$x$};
	\draw [->, thick] (0,0)  -- (0,11) node [above left, scale=1] {$P(x)$};
	\draw (10,0) node [below, scale=1] {$1$} -- (10,10);
	\draw (0,10) node [left, scale=1] {$1$} -- (10,10);
	\draw [thin] (0,0) -- (10,10);
	\draw (0,0.8) .. controls (0.4,0.8) .. (1,1) 
					.. controls (1.6,1) .. (2,2)
					.. controls (2.7,6) and (4,6) .. (6,6)
					.. controls (8,6.25) and (7.5,6) .. (8,8)
					.. controls (8.5,8.9)  .. (9.2,9.2)
					.. controls (9.6,9.3) .. (10,9.4);
	\draw [dashed] (1,1) -- (1,0) node [below, scale=1] {$x_{\min}$};
	\draw [dashed] (6,6) -- (6,0) node [below, scale=1] {$x^*$};
	\draw [dashed] (9.2,9.2) -- (9.2,0) node [below, scale=1] {$x_{\max}$};
\end{tikzpicture}
\par\end{centering}
\caption{\label{fig:Continuum-best-response}Continuum best response function
$P$}
\end{figure}
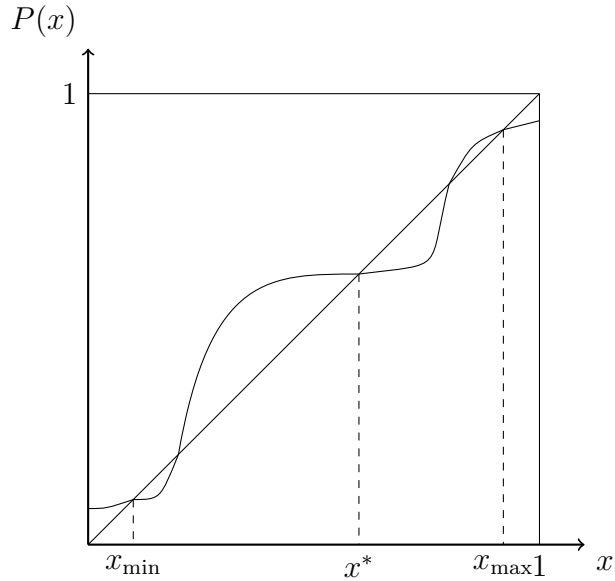

The goal of this paper is to study the set of all possible equilibrium
conventions or, more precisely, the set of equilibrium average actions.
Our results characterize the asymptotic upper and lower bounds \emph{in
the sense of set inclusion} on the equilibrium sets, across all networks.
Two results characterize the upper bound:
\begin{itemize}
\item Theorem \ref{thm:Complete graph} shows that if players live on a
sufficiently large complete graph, all stable fixed points of $P$
(essentially, fixed points where the graph of $P$ crosses the diagonal
from above) are arbitrarily close to average actions in some equilibrium.
(This and all subsequent results are stated ``with a probability
arbitrarily close to 1.'' ) That, generically, includes the largest
$x_{\max}$ and the smallest $x_{\min}$ fixed point of $P$. The
proof of Theorem \ref{thm:Complete graph} is straightforward. \\
A corollary to the Theorem shows that when players live on sufficiently
many disjoint copies of sufficiently large complete graphs, different
equilibria on component networks can be mixed and matched so that
the total average approximates arbitrary point on the interval $\left[x_{\min},x_{\max}\right]$
. 
\item Theorem \ref{thm:Largest equilibrium} shows that for all sufficiently
large and fine networks, there are no equilibria with average payoffs
above $x_{\max}$ or below $x_{\min}$. Although the statement is
very intuitive, our proof is surprisingly complicated. The difficulty
is to show that none of the profiles with average payoffs outside
of the range is an equilibrium. There are many such candidate profiles
and the claim must simultaneously address all of them. The difficulty
is compounded by the lack of additional assumptions on the network. 
\end{itemize}
Together, the two theorems show that the interval $\left[x_{\min},x_{\max}\right]$
is a tight upper bound on the sets of equilibrium average actions
across all networks. In this way, we obtain the strongest \emph{partial
identification} theory possible: without any further information about
the network, an econometrician who uses observed average behavior
$x$, can conclude that the parameters of the model must be such that
the parameter-dependent set $\left[x_{\min},x_{\max}\right]$ contains
$x$.

In particular, $x_{\min}=x_{\max}$ is a sufficient condition for
the existence of a unique heterogeneous equilibrium convention, regardless
of the network. As the subsequent results show, this condition is
not necessary for some networks. 

In order to characterize the lower bound on the equilibrium sets,
define a random utility-dominant, or $RU$-dominant, outcome $x^{*}$
as a solution to the maximization problem
\[
x^{*}\in\arg\max_{x}\intop_{0}^{x}\left(y-P^{-1}\left(y\right)\right)dy.
\]
(See Figure \ref{fig:Continuum-best-response}.) An $RU$-dominant
outcome is generically a stable fixed point of $P$. The notion of
RU-dominance is one of the contributions of this paper. When the impact
of payoff shocks on an individual utility converges to 0, the RU-dominant
outcome converges to the risk-dominant outcome (as in \cite{harsanyi_general_1988})
of the deterministic coordination $2\times2$ game. 

We have two results:
\begin{itemize}
\item Theorem \ref{thm:RUdominant selection on lattice} shows that there
exist networks where the average payoff in each equilibrium is arbitrarily
close to $x^{*}$. One example of such a network is a 2-dimensional
lattice. The idea of the proof is to show that for each profile with
an average behavior that is not $RU$-dominant, contagion-like best
response dynamics would bring the behavior close to $x^{*}$. The
proof uses an idea from \cite{morris_contagion_2000} to show how
a contagion wave spreads across lattice networks. This is supplemented
with explicit calculations of (a) the likelihood that a favorable
configuration of payoff shocks may initiate such a wave, and (b) the
likelihood that such a wave would not be stopped by an unfavorable
configuration of payoff shocks. The problem with the latter is the
reason why the 1-dimensional network of \cite{ellison_learning_1993}
is not a good example for the result and a 2-dimensional lattice is
needed.
\item Theorem \ref{thm:RU dominant always} shows that any sufficiently
large and fine network has an equilibrium with average payoffs close
to $x^{*}$. The starting point of the proof is a beautiful idea from
\cite{morris_contagion_2000} where it is shown that it is not possible
to spread risk-dominated actions by contagion. This idea is adapted
to work for $RU$-dominance, random utilities, etc. 
\end{itemize}
The two results together show that the single-element set $\left\{ x^{*}\right\} $
is a tight lower bound on all sets of equilibrium average payoffs
across all networks. This leads to an \emph{equilibrium selection}
theory: only outcome $x^{*}$ is robust to changes in the underlying
network.

Coordination games form one of three main approaches in the literature
that studies games in networks (\cite{jackson_games_2015}). The second
set of results of this paper is very closely related, and it greatly
benefits from the literature on contagion in networks, especially
from two beautiful papers, \cite{ellison_learning_1993} and \cite{morris_contagion_2000}.
\cite{ellison_learning_1993} (see also \cite{ellison_basins_2000})
was the first to show that a risk-dominant action can spread from
a small initial set of deviators to an entire 1-dimensional lattice
network by a simple best response process. \cite{morris_contagion_2000}
describes properties of networks for which Ellison's contagion wave
exists. Among others, any contagion wave from 1-dimensional lattices
can also be used in higher dimensions. \cite{morris_contagion_2000}
also shows that risk-dominated actions cannot spread through a best
response process no matter what is the geometry of the network. 

Evolutionary game theory (\cite{kandori_learning_1993}, \cite{young_evolution_1993},
\cite{blume_statistical_1993}, \cite{newton_conventions_2021}, and
many others) studies the long-run behavior of perturbed best response
processes, where players commit mistakes with a small probability,
and instead of choosing a best response, take some other action. One
of the key results of this literature is that the risk-dominant coordination
is (uniquely) stochastically stable regardless of the underlying network
(\cite{peski_generalized_2010}). Our current results (specifically,
Theorems \ref{thm:RUdominant selection on lattice} and \ref{thm:RU dominant always})
are closely related, but with some key differences. On the one hand,
there is a relation between ``noise'' in the behavioral rules of
the evolutionary literature and ``noise'' in the payoffs of the
current paper. On the other hand, there are two important differences:
We are interested here in static equilibria instead of a dynamic adjustment
process and our payoff shocks are permanent instead of temporary mistakes.
Finally, the evolutionary literature is subject to the criticism that
one may need to wait for a really long time before reaching a stochastically
stable outcome (\cite{ellison_learning_1993}). That criticism does
not apply to our model. 

Section \ref{sec:Model} contains the model. The next four sections
state and discuss the four theorems mentioned above. The last section
concludes. 

\section{Model\label{sec:Model}}

\subsection{Coordination game in a network}

There are $N$ agents $i=1,...,N$ who live in the nodes of a network.
The network is defined as an undirected weighted graph with weights
$g_{ij}=g_{ji}\geq0$ for $i,j\leq N$. We assume that $g_{ii}=0$
and that $g_{i}=\sum_{j}g>0$ for each player $i$. Let 
\begin{align*}
d\left(g\right) & =\max_{i,j}\frac{g_{ij}}{g_{i}}\text{ and }w\left(g\right)=\frac{\max_{i}g_{i}}{\min_{i}g_{i}},
\end{align*}
where $d\left(g\right)\in\left[0,1\right]$ is a bound on the importance
of a single player in another player's neighborhood and it describes
how fine the network is, and $w\left(g\right)\geq1$ is a rough measure
of the degree inequality. A network is\emph{ balanced} if all players
have the same degree $g_{i}=g_{j}$ for each $i,j$. In balanced networks,
$w\left(g\right)=1$. 

The agents play a binary action coordination game. Each agent chooses
an action $a_{i}\in\left\{ 0,1\right\} $ and receives a payoff 
\begin{equation}
\frac{1}{g_{i}}\sum_{j}g_{ij}u\left(a_{i},a_{-i},\varepsilon_{i}\right),\label{eq:payoffs}
\end{equation}
which depends on the actions of her neighbors and a payoff shock $\varepsilon_{i}\in\R$
drawn i.i.d. from a distribution $F\left(.\right)$. The payoffs are
supermodular in actions: for each $\varepsilon,$ 
\[
u\left(1,1,\varepsilon\right)+u\left(0,0,\varepsilon\right)>u\left(1,0,\varepsilon\right)+u\left(0,1,\varepsilon\right).
\]
Mixed actions are represented by the probability $a\in\left[0,1\right]$
of pure action 1. Due to expected utility, payoffs are linear in mixed
actions. We refer to the tuple $\left(u,F\right)$ as the \emph{random
utility game}. 
\begin{example}
\label{exa:Additive-payoff}In an additive payoff shock model, the
payoffs of player $i$ from interaction with $j$ are equal to 
\begin{equation}
u\left(a_{i},a_{j}\right)+\Lambda\varepsilon_{i}\mathbf{1}\left(a_{i}=1\right),\label{eq:model of rug}
\end{equation}
where $u$ is a symmetric $2\times2$ coordination game. Although
(\ref{eq:payoffs}) seems more general than (\ref{eq:model of rug}),
the two models are equivalent in the sense that the payoff shocks
can be matched so that the best responses to mixed strategies in both
models are identical. Parameter $\Lambda$ measures the importance
of the payoff shocks. When $\Lambda\rightarrow0$, the model converges
to the deterministic game. 
\end{example}

\subsection{Equilibria\label{subsec:Equilibria}}

We assume that the payoff shocks are publicly observable, i.e., players
know each others' preferences. Each network $g$, and each realization
of payoff shocks $\varepsilon$ leads to a many-player complete information
static game $G\left(g,\varepsilon\right)$. Let $\left(a_{i}\right)$
be a (possibly, mixed) profile of actions. Let
\begin{align*}
\text{Av}\left(a\right) & =\frac{1}{\sum_{i}g_{i}}\sum_{i}g_{i}a_{i}
\end{align*}
be the average action weighted by each player's neighborhood size.
This turns out to be the natural notion of average behavior. If $g_{i}\in\left\{ 0,1\right\} $,
then $g_{i}$ is a count of the interactions in which agent $i$ participates,
and $\text{Av}\left(a\right)$ is the average number of interactions
in which action 1 is played. 

Denote the set of average behaviors attained in Nash equilibria as
\[
\text{Eq}\left(g,\varepsilon\right)=\left\{ \text{Av}\left(a\right):a\text{ is a Nash eq. of }\text{G\ensuremath{\left(g,\varepsilon\right)}}\right\} \subseteq\left[0,1\right].
\]
$\text{Eq}\left(g\right)$ as a set-valued random variable, i.e.,
mapping from the space of payoff shock profiles to subsets of $\left[0,1\right]$.
The goal of the paper is to analyze the behavior of $\text{Eq}\left(g\right)$
as the network becomes larger and the importance of individual players
decreases, $d\left(g\right)\rightarrow0$. 

For any $x\in\left[0,1\right]$ and any two compact subsets $A,B\subseteq\left[0,1\right]$,
say \emph{$A$ is $\eta$-included in $B$}, write $A\subseteq_{\eta}B$,
if $\max_{x\in A}\min_{y\in B}\left|x-y\right|\leq\eta$. If $A\subseteq_{\eta}B$
and $B\subseteq_{\eta}A$, then we write $A=_{\eta}B$. 

\subsection{Continuum best response function}

For each payoff shock $\varepsilon$, define the best response threshold
$\beta\left(\varepsilon\right)$ as the fraction of people that would
make the player with payoff shock $\varepsilon$ indifferent between
the two actions:
\begin{equation}
u\left(1,\beta\left(\varepsilon\right),\varepsilon\right)=u\left(0,\beta\left(\varepsilon\right),\varepsilon\right).\label{eq:definition of beta(epsilon)}
\end{equation}
For each $x\in\left[0,1\right]$, let 
\[
P\left(x\right)=F\left(\beta\left(\varepsilon\right)\leq x\right).
\]
 $P\left(x\right)$ is the ex-ante probability that action 1 is a
best response if a player faces $x$ fraction of opponents who also
play 1. A typical graph of $P$ is illustrated on Figure \ref{fig:Continuum-best-response}.
The assumptions imply that $P$ is increasing, right-continuous, and
that $P\left(x\right)\in\left[0,1\right]$. We do not assume that
$P$ is invertible (and it won't be, if, for instance, $F$ has atoms).
Instead, we define $P^{-1}\left(y\right)=\inf\left\{ \left(x:P\left(x\right)\geq y\right)\right\} $. 

It is helpful to think about $P\left(x\right)$ as a best response
function in a continuum toy version of the game, where each agent's
payoff depends on the fraction of the entire population who choose
to play 1. Due to the continuum law of large numbers, $P\left(x\right)$
is the fraction of the population for whom $1$ is a best response.
Fixed points of $P$, i.e., intersections of the graph on Figure \ref{fig:Continuum-best-response}
with $45\lyxmathsym{\textdegree}$-line, correspond to Nash equilibria
in the continuum version of the game.

\section{Equilibria on complete graphs}

In this section, we consider a complete graph, i.e., network $g$
such that $g_{ij}=1$ for each $i\neq j$. For large $N$, such a
graph should approximate well the continuum toy model. 

We say that a fixed point $x=P\left(x\right)$ is\emph{ strongly stable}
if there exist $\gamma<1$ and a neighborhood $U\ni x$, such that
for each $y\in U$, if $y\leq x$ (resp. $y\geq x$), then $P\left(y\right)\geq P\left(x\right)+\gamma\left(y-x\right)$
(resp., $P\left(y\right)\leq P\left(x\right)+\gamma\left(y-x\right)$). 
\begin{thm}
\label{thm:Complete graph}Suppose that $x$ is a strongly stable
fixed point of $P$. Let $g^{N}$ be a complete graph with $N$ nodes.
For each $\eta>0$, there is $N>0$, such that 
\[
\Prob\left(\left\{ x\right\} \subseteq_{\eta}\text{Eq}\left(g^{N},\varepsilon\right)\right)\geq1-\eta.
\]
\end{thm}
Large complete graphs have equilibria that are close to strongly stable
points of $x$. The result is a sanity check, as it confirms our interpretation
of $P$ as a best response function on the continuum toy model. The
proof is straightforward (see Appendix \ref{sec:Proof-of-Complete graph}). 

When there are (finitely many) multiple strongly stable points, Theorem
\ref{thm:Complete graph} implies that, with a large probability,
all of them are close to the average behavior in some equilibrium.
In particular, if $x_{\min}$ and $x_{\max}$ are, respectively, the
smallest and the largest of the fixed points of $P$, then $\left\{ x_{\min},x_{\max}\right\} \subseteq_{\eta}\text{Eq}\left(g\right)$
with a large probability for a sufficiently large complete graph.

One can obtain other equilibrium averages by mixing and matching networks.
By taking a large number of disjoint copies of large complete graphs,
and considering a variety of equilibria on component networks, we
can an approximate an arbitrary point on the interval $\left[x_{\min},x_{\max}\right]$. 
\begin{cor}
\label{cor:upper bound}Suppose that $x_{\min}$ and $x_{\max}$ are
strongly stable. For each $\eta>0$, there exists a balanced network
$g$ such that 
\begin{align*}
\Prob\left(\left[x_{\min},x_{\max}\right]\subseteq_{\eta}\text{Eq}\left(g,\varepsilon\right)\right) & \geq1-\eta.
\end{align*}
\end{cor}

\section{Upper bound on equilibrium set}

The next result shows that $\left[x_{\min},x_{\max}\right]$ is an
upper bound on the equilibrium set. 
\begin{thm}
\label{thm:Largest equilibrium}Suppose that $x_{\min}$ and $x_{\max}$
are strongly stable. For each $\eta>0$ and $w<\infty$, there is
$\delta>0$ such that for each network $g$, if $d\left(g\right)\leq\delta$,
$w\left(g\right)\leq w$, then
\begin{align*}
\Prob\left(\text{Eq}\left(g\right)\subseteq_{\eta}\left[x_{\min},x_{\max}\right]\right) & \geq1-\eta.
\end{align*}
\end{thm}
The theorem yields a partial identification theory of the parameters
of the model. Consider an econometrician who studies a coordination
game on a network. The econometrician may not know the network $g$
on which the game is played, nor the parameters of the random utility
model, and she treats them as parameters. If she observes the average
behavior $x$, she may reject all parameters for which $x\notin\left[x_{\min},x_{\max}\right]$. 

Theorem \ref{thm:Largest equilibrium} and Corollary \ref{cor:upper bound}
together show that the interval $\left[x_{\min},x_{\max}\right]$
is a tight upper bound (in the sense of set inclusion) on the average
behavior across all networks. In particular, the partial identification
obtained from the result cannot be improved. 

\subsection{Proof intuition}

Our proof of Theorem \ref{thm:Largest equilibrium} is surprisingly
complicated. To explain the difficulty, fix an average payoff $x>x_{\max}+\eta$.
For each profile $a$ such that $\text{Av}\left(a\right)\geq x$,
it is relatively easy to show that the ex ante probability that $a$
is an equilibrium is small. In fact, one can bound this probability
with an exponential bound 
\begin{equation}
\leq\exp\left(-\delta_{\eta}N\right)\label{eq:simple prob bound}
\end{equation}
where $\delta_{\eta}>0$ may depend on the geometry of the network,
etc. Importantly, if $\eta$ is very small, the bound constant $\delta_{\eta}$
is very small as well. The idea is that if the average action is above
the largest fixed point, than a relatively large number of players
must be best responding significantly above the continuum best response
function, which cannot happen with a significant probability. 

The above bound applies to a particular profile $a$. In order to
obtain a bound for all profiles $a$ such that $\text{Av}\left(a\right)\geq x$,
we can multiply (\ref{eq:simple prob bound}) by the number of such
profiles. Unfortunately, this number of order 
\[
\exp\left(\left(x\log x+\left(1-x\right)\log\left(1-x\right)\right)N\right),
\]
and, if $\delta_{\eta}$ is sufficiently small, or if $x$ is sufficiently
close to $x_{\max}$, it converges to infinity much faster than (\ref{eq:simple prob bound})
converges to 0. 

In the proof, we divide the profiles $a$ into groups such that (a)
we can show that (\ref{eq:simple prob bound}) is an upper bound on
the probability that none of the profiles in a group is an equilibrium
(Lemma \ref{lem:Uper equilibrium} in the Appendix), and (b) the number
of groups grows at a much slower rate than (\ref{eq:simple prob bound})
decreases (Lemma \ref{lem:Sudakov} in the Appendix). 

The idea of the division comes from an observation that differences
between profiles matter for a player $i$ only if they lead to different
distributions of actions among neighbors of $i$. Formally, for each
profile $a$, construct a profile $\beta^{a}$ of average neighborhood
behaviors so that for each $i$, $\beta_{i}^{a}=\frac{1}{g_{i}}\sum_{j}g_{ij}a_{j}$.
Then, if $a$ is an equilibrium, it must be that $a_{i}\geq1$ if
and only if $\beta\left(\varepsilon_{i}\right)\leq\beta_{i}^{a}$
for each $i$. We define a notion of closeness of two profiles of
neighborhood behaviors as a weighted version of the Euclidean metric:
\[
d\left(\beta^{a},\beta^{b}\right)=\sqrt{\frac{1}{\sum g_{i}^{2}}\sum g_{i}^{2}\left(\beta_{i}^{a}-\beta_{i}^{b}\right)^{2}}.
\]

Property a) is a consequence of the observation that if two profiles
generate similar distributions of neighbor actions for all, or at
least for a great majority of players, they should lead to similar
best responses. Hence the question of such profiles being equilibria
is highly correlated, which makes it easier to ensure that a quantity
like (\ref{eq:simple prob bound}) provides a bound that no profile
in the entire group is an equilibrium. For step (b), let $\mathcal{B}=\left\{ \beta^{a}:a\text{ is a profile}\right\} $
be the set of all neighborhood behavior profiles. We show that the
number of balls of radius $\delta$ (in metric $d$) that is required
to cover this set, i.e., the metric entropy of $\mathcal{B}$, is
of order 
\[
\exp\left(\delta_{d\left(g\right),\delta}^{\prime}N\right),
\]
where $\delta_{d\left(g\right)}\rightarrow0$ as $d\left(g\right)\rightarrow0$.
In particular, when $d\left(g\right)$ is sufficiently small, i.e.,
no player dominates the neighborhood of another player, the above
bound converges to infinity at much slower rate than (\ref{eq:simple prob bound})
converges to 0. 

\section{RU-dominant selection}

In this section, we introduce an equilibrium selection tool appropriate
for coordination games with random utility: the random utility-, or
$RU$-dominant outcome. We show that there are networks on which the
$RU$-dominant outcome is essentially the only equilibrium average.

\subsection{$RU$-dominant outcome}

An equilibrium action $x^{*}\in\left[0,1\right]$ is \emph{RU-dominant
}if it is a maximizer of 
\begin{equation}
x^{*}\in\arg\max_{x}\intop_{0}^{x}\left(y-P^{-1}\left(y\right)\right)dy.\label{eq:RU-dominant}
\end{equation}
It is \emph{strictly RU-dominant}, if it is a unique maximizer. Generically,
any game with random utility has an \emph{RU}-dominant action. 

The following example shows that if the impact of the random utility
impact disappears, the RU-dominant outcomes converge to standard risk
dominance of \cite{harsanyi_general_1988}.
\begin{example}
\label{exa:Limit additive risk-dominant}(Cont. of Example \ref{exa:Additive-payoff})
Suppose w.l.o.g. that $0$ is the unique strictly risk-dominant action
of the coordination game with payoffs $u$. Then, each player is indifferent
between two actions if a fraction $\alpha>\frac{1}{2}$ of players
plays action $1$. When $\Lambda\rightarrow0$, $P^{-1}\left(y\right)\rightarrow\alpha$
for each $y\in\left(0,1\right)$, and we have 
\[
\intop_{0}^{x}\left(y-P^{-1}\left(y\right)\right)dy\rightarrow\intop_{0}^{x}\left(y-\alpha\right)dy=\frac{1}{2}x^{2}-\alpha x=x\left(\frac{1}{2}-\alpha\right).
\]
Hence the RU-dominant outcome(s) converge to 0, i.e., the risk-dominant
action of deterministic game $u$. 
\end{example}
The main result of this section shows that there are networks where,
with a large probability, all equilibrium averages are close to a
strictly $RU$-dominant outcome $x^{*}$.
\begin{thm}
\label{thm:RUdominant selection on lattice}Suppose that $x^{*}$
is the strictly $RU$-dominant outcome and that either $x^{*}>0$
and $P\left(0\right)>0$, or $x^{*}<1$ and $P\left(1\right)<1$.
For each $\eta>0$, there is a network $g$ such that 
\begin{align*}
\Prob\left(\text{Eq}\left(g\right)\subseteq_{\eta}\left\{ x^{*}\right\} \right) & \geq1-\eta.
\end{align*}
\end{thm}

\subsection{Proof intuition\label{subsec:Proof-intuition}}

The network constructed in the proof is a $2$-dimensional lattice,
parameterized with $M$ and $m$. There are $M^{2}$ agents located
on a square $\left[0,\frac{M}{m}\right]^{2}\subseteq\R^{2}$ at fractional
points of form $\left(\frac{k}{m},\frac{l}{m}\right)$ for some $k,l=1,...,M$.
The two agents are connected, $g_{ij}=1$, if the (Euclidean) distance
between them is no larger than $1$. (To make the network balanced
and to simplify the argument, we assume that all distance calculations
are done $\text{mod}\frac{M}{m}$, which turns the square $\left[0,\frac{M}{m}\right]^{2}$
into a torus.) The proof requires both $m$ and $\frac{M}{m}$ to
be sufficiently large. Our argument and the result extends to $K$-dimensional
lattices for $K>2$, but not to $K=1$.

The proof has three key steps. In order to illustrate the first two,
consider a version of the line network from \cite{ellison_learning_1993}.
Agents are located along a line at equally spaced and dense locations
and the weight of connection between agents $i$ and $j$ depends
only on their distance $g_{ij}=g_{i-j}=g_{j-i}$. We normalize the
weights so that $\sum g_{d}=1$. We are going to show that there cannot
be an equilibrium with average actions substantially higher than $x^{*}$.
(An analogous argument shows that there cannot be an equilibrium with
actions substantially lower than $x^{*}$.) Suppose to the contrary
that there is. The first step is to notice that if the line is sufficiently
long, then, with a probability close to 1, there will be a group of
consecutive agents with payoff shocks that render action 0 dominant.
We refer to them as the initial infectors. 

The second step it to show that, starting from the initial infectors,
a contagion best response process will spread across all ``good''
agents to bring their actions down below $x^{*}$, where a group of
agents is ``good'' if the empirical distribution of payoff shocks
in the group is close to $F$. To simplify the argument, assume that
each location in the network contains a continuum population; the
law of large numbers implies that the average equilibrium action of
agents in node $i$ is equal to $P\left(\sum_{d}g_{d}a_{i+d}\right)$.
Suppose that all locations $i\leq0$ consist of initial infectors
and have average actions not higher than $x^{*}$. Assume that, initially,
all locations $i>0$ play action 1. Consider a best response process
where each location $i>0$ is allowed to revise its average action
to its best response, but not less than $x^{*}$. The process will
either end with all locations playing $x^{*}$, or the contagion will
stop. Suppose the latter. Let $a_{i}\geq x^{*}$ be the final average
action in location $i$. Due to payoff complementarities, $a_{i}$
must be increasing in $i$. Let $a=\lim_{i\rightarrow\infty}a_{i}>x^{*}$.
Because the best response process stopped for each location $i$,
if $a_{i}>x^{*}$, the average action must be equal to the best response
action
\[
a_{i}=P\left(\sum_{d}g_{d}a_{i+d}\right).
\]
Taking inverse, we obtain
\[
P^{-1}\left(a_{i}\right)\leq\sum_{d}g_{d}a_{i+d}=x^{*}+\sum_{j}\left(\sum_{d\geq j-i}g_{d}\right)\left(a_{j+1}-a_{j}\right),
\]
where we use $a_{i}\geq x^{*}$, and the equality is due to a discrete
version of the ``integration by parts'' formula. After subtracting
$x^{*}$, multiplying by $a_{i+1}-a_{i}\geq0$, and summing up across
locations $x^{*}$ gives
\begin{equation}
\sum_{i}\left(P^{-1}\left(a_{i}\right)-x^{*}\right)\left(a_{i+1}-a_{i}\right)\leq\sum_{i,j}\left(\sum_{d\geq j-i}g_{d}\right)\left(a_{i+1}-a_{i}\right)\left(a_{j+1}-a_{j}\right).\label{eq:inequality contagion}
\end{equation}
The left-hand side of the inequality is approximately equal to $\intop_{x^{*}}^{a}\left(P^{-1}\left(y\right)-x^{*}\right)dy$.
To compute the right hand side, notice that we can switch the roles
of $i$ and $j$ in the summation, and using the fact that $\sum_{d\geq j-i}g_{d}+\sum_{d\geq i-j}g_{d}=\sum g_{d}=1$,
we have
\begin{align*}
\sum_{i,j}\left(\sum_{d\geq j-i}g_{d}\right)\left(a_{i+1}-a_{i}\right)\left(a_{j+1}-a_{j}\right) & =\frac{1}{2}\left(\sum_{i,j}\left(\sum_{d\geq j-i}g_{d}+\sum_{d\geq i-j}g_{d}\right)\left(a_{i+1}-a_{i}\right)\left(a_{j+1}-a_{j}\right)\right)\\
 & =\frac{1}{2}\left(\sum_{i,j}\left(a_{i+1}-a_{i}\right)\left(a_{j+1}-a_{j}\right)\right)=\frac{1}{2}\left(a-x^{*}\right)^{2}\\
 & =\frac{1}{2}\left(a-x^{*}\right)^{2}=\intop_{x^{*}}^{a}\left(y-x^{*}\right)dy.
\end{align*}
Putting the two sides together, inequality (\ref{eq:inequality contagion})
implies that 
\[
\intop_{x^{*}}^{a}\left(y-P^{-1}\left(y\right)\right)dy\geq0,
\]
which contradicts the fact that $x^{*}$ is the unique maximizer of
the integral on the right-hand side. Hence the contagion wave must
spread across the entire network. 

The third step is to make sure that the contagion wave is not stopped
by ``bad'' agents, whose preference shocks are more favorable towards
higher actions, or perhaps even turn action 1 into a dominant action.
A consecutive set of ``bad'' agents will not revise down their actions
in a way described by the above equations and, if sufficiently large,
may block the contagion wave from moving over them. Because a set
of ``bad'' agents has a positive probability, it is important to
compare the relative frequency of the sets of initial infectors necessary
to start the wave versus the sets of ``bad'' agents who may stop
it. Unfortunately, for some $P$s, the latter are more frequent on
the line network. As a result, the line network is not a good candidate
example for Theorem \ref{thm:RUdominant selection on lattice}.

The spread of a contagion wave from a small set of initial infectors
extends from the line to higher-dimensional lattices due to an elegant
argument from \cite{morris_contagion_2000}. The idea is that if the
front of the wave is sufficiently smooth, then it can be locally approximated
by a hyperplane; its spread in the orthogonal direction is up to some
approximations identical to the spread along a one-dimensional lattice. 

At the same time, the existence of ``bad'' sets is less of an issue
on higher-dimensional lattices. The reason is that, even if ``bad''
sets are more likely then the sets of initial infectors, in order
for them to stop the wave, they would have to be arranged so to surround
the initial infectors. We show that the likelihood of such arrangement
is very low, and if $m$ and $\frac{M}{m}$ are sufficiently large,
it is much lower than the likelihood of the set of initial infectors.
Although this observation is intuitive, a rigorous proof is lengthy
and relies on some ideas from percolation theory (\cite{bollobas_percolation_2006}).
More precisely, the proof surrounds each ``bad'' set with an open
ball of large but fixed radius. We show that, for large $\frac{M}{m}$,
the size of all such balls is small relative to the size of the network,
and that the rest of the network has a giant connected component (i.e.,
component that contains a fraction almost equal to 1 of all agents
in the network and such that all agents are connected). If the radius
of the balls isolating the ``bad'' sets is sufficiently large, we
show we can construct sufficiently smooth contagion wave. The details
are left to Appendix \ref{sec:Proof-of-Unique selection}.

\section{RU-dominant equilibrium in each network }

The previous section identified an $RU$-dominant outcome as a candidate
solution for equilibrium selection theory. Next, we ask whether there
are other potential candidates, i.e., whether there are other outcomes
that can be unique equilibria on some networks. 

The next result shows that the answer is negative.
\begin{thm}
\label{thm:RU dominant always}Suppose that $x^{*}$ is the strictly
$RU$-dominant outcome. For each $\eta>0$, there is $d>0$ such that,
for each network $g$, if $d\left(g\right)\leq d$, then 
\begin{align*}
\Prob\left(\left\{ x^{*}\right\} \subseteq_{\eta}\text{Eq}\left(g\right)\right) & \geq1-\eta.
\end{align*}
\end{thm}
If the network is sufficiently fine, then, for almost all realizations
of payoff shocks, there is an equilibrium with action distribution
close to the RU-dominant action. In particular, no other outcome than
the $RU$-dominant outcome can be a unique equilibrium in some network. 

Theorems \ref{thm:RUdominant selection on lattice} and \ref{thm:RU dominant always}
lead to an \emph{equilibrium selection} theory: only the $RU$-dominant
outcome $x^{*}$ is robust to changes in the underlying network. This
claim is made precise by the proof of Theorem \ref{thm:RU dominant always}.
In the proof, we consider a profile in which almost all players choose
best responses as if $x^{*}$ neighbors play action 1. We show that
any best response dynamics starting from such a profile will stop
in an equilibrium profile in which a great majority of players never
revise their actions. It follows that, if players play such an equilibrium
under one network, and then the network is changed (in a manner independent
of actions and payoff shocks), then the best response process will
end up with a very similar profile as an equilibrium. 

\subsection{Proof intuition}

We start with an initial profile $a^{0}$ in which all players choose
best responses as if fraction $x^{*}$ of their opponents plays 1,
\[
a_{i}^{0}\in\arg\max u\left(a_{i},x^{*},\varepsilon_{i}\right).
\]
Although each agent chooses depending on their payoff shock, the law
of large numbers and the fact that $x^{*}$ is an equilibrium of the
continuum game imply that the average action in the population is
unlikely to be far from $x^{*}$.

Starting from the initial profile, we consider an \emph{upper} best
response dynamics, where at each stage, a single player is allowed
to revise their action towards the best response, but only upwards,
i.e., if the best response is the action $1$. Such dynamics must
stop eventually, and the resulting profile $a^{U}$ does not depend
on the order in which players revise their actions, as long as all
players for whom 1 is the best response has the opportunity to revise.
We argue below that the average action under $a^{U}$ is not too far
from the average action under $a^{0}$, and hence from $x^{*}$. Similarly,
an analogous observation holds when we analyze a downward counterpart
of the best response dynamics. Because of payoff complementarities,
there must be an equilibrium action profile sandwiched between the
limit profiles obtained by the upward and downward best response dynamics.
The two observations imply that such an equilibrium is not far away
from $x^{*}$. 

In order to explain the key observation, it is helpful to begin with
a special case of Example \ref{exa:Additive-payoff}, or when the
game is close to being deterministic and $x^{*}$ is close to 0. In
this case, Theorem \ref{thm:RU dominant always} follows from an argument
that based on the proof of Proposition 3 in \cite{morris_contagion_2000}.
(S. Morris attributes this idea to D. McAdams.) Let $a^{t}$ be the
$t$th stage of the upward best response dynamics. At each stage,
we define the infection capacity of profile $a^{t}$ as the mass of
links that connect agents who play action 1 with agents who play action
0, 
\begin{equation}
\mathcal{F}_{0}\left(a\right)=\sum_{i,j:a_{i}^{t}=1,a_{j}^{t}=0}g_{ij}.\label{eq:simple capacity}
\end{equation}
If, at stage $t+1$, player $i$ revises her action upwards, then
(a) the capacity will increase by $\sum_{j:a_{j}^{t}=0}g_{ij}$ because
of her new out-going links, and (b) it will decrease by $\sum_{j:a_{j}^{t}=1}g_{ij}$,
i.e., by the weight of the links from player $i$ to others who choose
1 in profile $a^{t}$. Because action $1$ is a best response of player
$i$, assuming that player $i$'s payoffs are close to deterministic
utility $u$, it must be that 
\[
\sum_{j:a_{j}^{t}=1}g_{ij}\approx\alpha\sum_{j}g_{ij}=\alpha g_{i}\text{ and }\text{\ensuremath{\sum_{j:a_{j}^{t}=0}g_{ij}\approx\left(1-\alpha\right)}}g_{i}.
\]
(Recall that $\alpha>\frac{1}{2}$ is a fraction of neighbors that
makes players indifferent between two actions.) Hence the capacity
in stage $t+1$ will be $\alpha g_{i}-\left(1-\alpha\right)g_{i}=\left(2\alpha-1\right)g_{i}$
smaller. Because the capacity cannot fall below 0, this leads to a
bound on the total mass of players who switch action under the dynamics
\[
\left(2\alpha-1\right)\sum_{i:0=a_{i}^{0}<a_{i}^{U}=1}g_{i}\leq\mathcal{F}_{0}\left(a^{0}\right).
\]
Because the initial profile was close to $0$, the capacity and the
limit profile must be close to 0 as well. Hence the number of agents
who revise their actions is small.

There are two important features of the above argument: the initial
capacity is small and it must appropriately decrease with each action
revised upward. The proof of Theorem \ref{thm:RU dominant always}
preserves the two features, but with a modified notion of capacity.
We cannot use (\ref{eq:simple capacity}), because, for general payoff
shocks and $x^{*}\in\left(0,1\right)$, a substantial fraction of
the population plays each action and (\ref{eq:simple capacity}) is
too large. Instead, we replace actions $a_{i}$ by their expected
best response versions $p_{i}=P\left(\frac{1}{g_{i}}\sum_{j}g_{ij}a_{j}\right)$
and define
\begin{equation}
\mathcal{F}\left(p\right)=\frac{1}{2}\sum_{i,j}g_{ij}\left(p_{i}-p_{j}\right)^{2}.\label{eq:real capacity}
\end{equation}
(To motivate the definition, notice that if we replace $p_{i}$ by
$a_{i}$, then (\ref{eq:simple capacity}) and (\ref{eq:real capacity})
are equal.) 

The law of large numbers implies that, under the initial profile $a^{0}$,
the average action among the neighbors, $\beta_{i}^{0}=\frac{1}{g_{i}}\sum_{j}g_{ij}a_{j}^{0}$,
and hence the expected best response $p_{i}$ must also be close to
$x^{*}$. Thus, the capacity of the initial profile is appropriately
small and the first required feature of capacity is preserved. 

The second feature is preserved as well. We sketch the idea here and
leave the details to the Appendix. Due to symmetry in the weights
$g_{ij}=g_{ji}$ for each $i,j$, we have for each $t$,
\[
\mathcal{F}\left(p^{t+1}\right)-\mathcal{F}\left(p^{t}\right)=\sum_{i}g_{i}\left(p_{i}^{t+1}\right)^{2}-\sum_{i}g_{i}\left(p_{i}^{t}\right)^{2}-\sum_{i}\left(p_{i}^{t+1}-p_{i}^{t}\right)\sum_{j}g_{ij}\sum_{s=t,t+1}p_{j}^{s}.
\]
Summing across $t\leq T$, and letting $\beta_{j}=\frac{1}{g_{j}}\sum_{i}g_{ij}a_{i}$
be the average behavior of $j$'s neighbors, some algebra shows that
\begin{align}
 & \mathcal{F}\left(p^{T+1}\right)-\mathcal{F}\left(p^{0}\right)=\sum_{t\leq T}\left(\mathcal{F}\left(p^{t+1}\right)-\mathcal{F}\left(p^{t}\right)\right)\nonumber \\
= & -2\sum_{i}g_{i}\left[\intop_{p_{i}^{0}}^{p_{i}^{T+1}}\left(P^{-1}\left(y\right)-y\right)dy\right]+\sum_{t\leq T}\sum_{i}\left(p_{i}^{t+1}-p_{i}^{t}\right)\sum_{j}g_{ij}\sum_{s=t,t+1}\left(a_{j}^{s}-p_{j}^{s}\right)\label{eq:computations risk-dominant always}\\
 & +\text{ small terms.}\nonumber 
\end{align}
 The details of the calculations can be found in Appendix \ref{sec:Proof-of-AllDominant}.
The ``small terms'' depend on the stage increase in $\beta_{i}^{t+1}-\beta_{i}^{t}$,
which is small due to our assumption that at most one agent revises
her action per period and because the impact of a single agent in
the neighborhood of another is smaller than $d\left(g\right)$. They
also depend on the difference $\beta_{i}^{0}-x^{*}$, which is small
because the initial profile is close to $x^{*}$.

The second term of the right-hand side is small for probabilistic
reasons. Notice that the probability that action $1$ is a player's
$j$ best response in period $s$ is not higher than the expected
action $p_{j}^{s}$. In fact, the probability is not higher even it
is conditioned on the actions of other agents. Theis that agent $j$'s
behavior positively affects the actions of other players only after
she revises her action. This observation, together with the fact that
each agent $j$ is small in the neighborhood of $i$, allows us to
show that the second last term is small, with a large probability,
due to a version of the finite law of large numbers. 

Ignoring the (probabilistically and deterministically) small terms,
summing across $t$, and remembering that $p_{i}^{t}=P\left(\beta_{i}^{t}\right)$
and that $\beta_{i}^{0}\approx x^{*}=P\left(x^{*}\right)\approx P\left(\beta_{i}^{0}\right)$,
we obtain
\begin{align*}
\mathcal{F}\left(p^{0}\right)\geq & 2\sum_{i}g_{i}\left[\intop_{x^{*}}^{P\left(\beta_{i}^{T}\right)}\left(P^{-1}\left(y\right)-y\right)dy\right].
\end{align*}
The definition of the RU-dominant outcome implies that, at least locally,
the integral is increasing in $\beta_{i}^{T}$. Hence, if the original
capacity is small, then, for each $T$, the average behavior in the
neighborhood of a great majority of players cannot be too far away
from $x^{*}$. Hence, the limit of the upper best response dynamics
cannot be to far away from $x^{*}$, which concludes the argument. 

\section{Discussion}

\subsection{Unweighted average\label{subsec:Average-action}}

Our definition of the average action stated in Section \ref{subsec:Equilibria}
weights individuals by their neighborhood size $g_{i}$. An alternative
is to use the unweighted average 
\begin{align*}
\text{Av}_{\text{unweighted}}\left(a\right) & =\frac{1}{N}\sum_{i}a_{i}.
\end{align*}
When the network is balanced, i.e., when $g_{i}=g_{j}$ for each $i$
and $j$, the two notions of average are identical. 

Because Theorem \ref{thm:Complete graph}, Corollary \ref{cor:upper bound},
and Theorem \ref{thm:RUdominant selection on lattice} are proven
using balanced networks, they continue to hold \emph{verbatim} if
we change the notion of average to unweighted one. A version of Theorem
\ref{thm:RU dominant always} holds with the following modification
:\emph{ for each $w>0$, and each $\eta>0$, there is $d>0$ such
that, for each network $g$, if $d\left(g\right)\leq d$ and $w\left(d\right)\leq w$,
then 
\[
\Prob\left(\left\{ x^{*}\right\} \subseteq_{\eta}\text{Eq}\left(g,\right)\right)\geq1-\eta.
\]
}The required modification of the proof is very minor and it can be
found in Appendix \ref{sec:Extension-to-unweighted}. 

We were not able to find an immediate way of extending Theorem \ref{thm:Largest equilibrium}. 

\subsection{Small number of links}

The results of this paper focus on the limit case $d\left(g\right)\rightarrow0$
and they apply to networks with a large number of connections (i.e.,
large degrees), like networks of acquaintances. If $d\left(g\right)>0$,
none of the results hold. The small-degree case requires different
techniques and separate analysis and we leave it for future research. 

\subsection{Independence}

Another key assumption of the model is that the payoff shocks are
independent across agents. An alternative and natural assumption is
that the payoff shocks of directly connected agents can be correlated.
If imperfect, such a correlation dies out exponentially with the distance
between agents, making distant agents roughly independent. For this
reason, we suspect that the results of this paper continue to hold.
However, the proper analysis of this case is left to future research. 

\appendix

\section{Monotonicity}

This part of the Appendix shows that if $P$ is a continuum best response
function of random utility game $\left(u,F\right)$, then, for any
increasing and right-continuous function $P^{\prime}\geq P$, there
is a random utility game that has $P^{\prime}$ as a continuum best
response function, and such that the distribution of equilibria first-order
stochastically dominates the distribution of equilibria in the original
game. 

Formally, the space of (mixed) action profiles $\mathcal{A}=\left[0,1\right]^{N}$
is a lattice with coordinate-wise comparison: for any $a,b\in\mathcal{A}$,
we have $a\leq b$ iff $a_{i}\leq b_{i}$ for each $i$. Let $\leq_{S}$
denote the strong set order on subsets or $\R$ and, as a lattice
extension, of $\mathcal{A}$. We say that a probability distribution
$\mu\in\Delta\mathcal{A}$ is dominated by $\mu^{\prime}\in\Delta\mathcal{A}$
in the sense of first-order stochastic dominance, and write $\mu\leq_{FOS}\mu^{\prime}$,
if for each $a\in\mathcal{A}$, $\mu\left(\left\{ a^{\prime}:a^{\prime}\geq a\right\} \right)\leq\mu^{\prime}\left(\left\{ a^{\prime}:a^{\prime}\geq a\right\} \right)$.

Let $\left(u,F\right)$ be a random utility game. Let $E\left(u,\varepsilon\right)$
denote the set of equilibrium profiles in random utility game $\left(u,F\right)$.
We compare sets using the strong set order. Let $\mu\left(u,F\right)$
denote the probability distribution over the sets of equilibrium profiles
induced by distribution over profiles of payoffs shocks. We say that
random utility game $\left(u,F\right)$ is dominated by game $\left(u^{\prime},F^{\prime}\right)$
if $\mu\left(u,F\right)\leq_{FOS}\mu\left(u^{\prime},F^{\prime}\right)$. 
\begin{lem}
\label{lem:Dominance}Suppose that $P$ is a continuum best response
function of random utility game $\left(u,F\right)$. Then, for each
increasing, right-continuous $P^{\prime}\geq P$, there exists random
utility game $\left(u^{\prime},F^{\prime}\right)$ such that (a) $P^{\prime}$
is a continuum best response function of $\left(u^{\prime},F^{\prime}\right)$,
and (b) random utility game $\left(u,F\right)$ is dominated by game
$\left(u^{\prime},F^{\prime}\right)$.
\end{lem}
\begin{proof}
First, observe that any two random utility models with the same continuum
best response function $P$ have the same distributions over sets
of equiluibria. Second, we show that we can construct different models
over the same probability space. Let $\Omega=\left[0,1\right]^{N}$
and let $\lambda^{N}$ be the product uniform measure on $\Omega$.
For each increasing, right-continuous $P$, define utility function
$u_{P}$ so that 
\[
u_{P}\left(1,x,\varepsilon\right)=x-P^{-1}\left(\varepsilon\right)\text{ and }u_{P}\left(0,x,\varepsilon\right)=0.
\]
Then, the continuum best response function of $\left(u_{P},\lambda\right)$
is equal to $P$. Finally, notice that if $P^{\prime}\geq P$, and
we consider two games $u_{P}$ and $u_{P}^{\prime}$ on the same probability
space $\left(\Omega,\lambda^{N}\right)$, then the best resposne of
each player in the second model is always higher (in the sense of
strong set order) than the best response of the player in the first
model. A consequence is that, for each $\varepsilon,$ $E\left(u_{P},\varepsilon\right)\leq E\left(u_{P^{\prime}},\varepsilon\right)$,
which concludes the proof of the result. 
\end{proof}

\section{Proof of Theorem \ref{thm:Complete graph} and Corollary \ref{cor:upper bound}\label{sec:Proof-of-Complete graph}}

\subsection{Proof of Theorem \ref{thm:Complete graph}}

Let $U$ be an open set from the definition of a strongly stable $x$.
Fix $\delta>0$ and $N<\infty$ such that $\left[x-2\delta,x+2\delta\right]\subseteq U$
and $\gamma\frac{1}{N-1}\leq\frac{1}{2}\left(1-\gamma\right)\delta$.
Let $\eta=\frac{1}{2}\left(1-\gamma\right)\delta$. Then, 
\[
x-\delta\leq x-\delta+\left(\left(1-\gamma\right)\delta-\gamma\frac{1}{N-1}\right)-\eta\leq P\left(x\right)-\gamma\left(\delta+\frac{1}{N-1}\right)-\eta\leq P\left(x-\delta-\frac{1}{N-1}\right)-\eta,
\]
and similarly, $P\left(x+\delta+\frac{1}{N-1}\right)+\eta\leq x+\delta$.
Additionally, choose a sufficiently large $N$ so that $2\exp\left(-2N\eta^{2}\right)\leq\eta$.

Let 
\[
P_{\varepsilon}\left(x\right)=\frac{1}{N}\sum_{i}\mathbf{1}\left(\beta\left(\varepsilon_{i}\right)\leq x\right),
\]
be the empirical distribution of best response thresholds. Define
event $\mathcal{P}=\left\{ \sup_{x}\left|P_{\varepsilon}\left(x\right)-P\left(x\right)\right|\leq\eta\right\} $.
By the Dvoretsky-Kiefer--Wolfowitz--Massart inequality, for each
$\eta>0$,
\[
\text{Prob}\left(\text{not }\mathcal{P}\right)\leq2\exp\left(-2N\eta^{2}\right)\leq\eta.
\]

For each profile $a$, define $\beta_{i}^{a}=\frac{1}{N-1}\sum_{j\neq i}a_{j}$
as the average action in player $i$'s neighborhood. The average action
is not far from the average action in the population, $\left|\beta_{i}^{a}-\text{Av}\left(a\right)\right|\leq\frac{1}{N-1}.$ 

Suppose that event $\mathcal{P}$ holds. Let $b\left(a,\varepsilon\right)$
be the best response profile to profile $a$, where, in a case of
a tie, we assume that an agent chooses $1$. Then, 
\[
\text{Av}\left(b\left(a,\varepsilon\right)\right)=\frac{1}{N}\sum\mathbf{1}\left\{ \beta\left(\varepsilon_{i}\right)\leq\beta_{i}^{a}\right\} .
\]
If $\text{Av}\left(a\right)\in\left[x-\delta,x+\delta\right]$, the
above inequalities imply that 
\begin{align*}
 & x-\delta\\
\leq & P\left(\text{Av}\left(a\right)-\frac{1}{N-1}\right)-\eta\leq\frac{1}{N}\sum\mathbf{1}\left\{ \beta\left(\varepsilon_{i}\right)\leq\text{Av}\left(a\right)-\frac{1}{N-1}\right\} \\
\leq & \text{Av}\left(b\left(a,\varepsilon\right)\right)\\
\leq & \frac{1}{N}\sum\mathbf{1}\left\{ \beta\left(\varepsilon_{i}\right)\leq\text{Av}\left(a\right)+\frac{1}{N-1}\right\} \leq\eta+P\left(\text{Av}\left(a\right)+\frac{1}{N-1}\right)\\
\leq & x-\delta.
\end{align*}
Hence, mapping $b\left(,.\varepsilon\right)$ maps the set of profiles
$a$ s.t. $\text{Av}\left(a\right)\in\left[x-\delta,x+\delta\right]$
into itself. The result follows from the fixed-point theorem. 

\subsection{Proof of Corollary \ref{cor:upper bound}}

By Theorem \ref{thm:Complete graph}, for each $\eta>0$ and for sufficiently
large $N$, 
\[
\Prob\left(\left\{ x_{\min},x_{\max}\right\} \subseteq_{\frac{1}{8}\eta}\text{Eq}\left(g^{N},\varepsilon\right)\right)\leq\sum_{x\in\left\{ x_{\min},x_{\max}\right\} }\Prob\left(\left\{ x\right\} \subseteq_{\frac{1}{8}\eta}\text{Eq}\left(g^{N},\varepsilon\right)\right)\leq\frac{1}{4}\eta.
\]

Let $g=g^{K,N}$ be a balanced network that consists of $K$ copies
of complete $N$-person graphs. Let $g_{k}$ denote the $k$th copy.
Let $A\left(g,\varepsilon\right)=\left\{ k:\left\{ x_{\min},x_{\max}\right\} \subseteq_{\frac{1}{8}\eta}\text{Eq}\left(g_{k},\varepsilon\right)\right\} $
be the set of copies that contain equilibria with averages close to
the largest and the smallest of the fixed points. By the choice of
$N$ and the Central Limit Theorem, for sufficiently large $K$,
\[
\text{Prob}\left(\frac{1}{K}\left|A\left(g,\varepsilon\right)\right|\leq1-\frac{3}{8}\eta\right)\leq\eta.
\]

Let $\psi_{\max},\psi_{\min}:\left\{ 1,...,K\right\} \rightarrow\left[0,1\right]$
be functions such that $\psi_{s}\left(k\right)\in\text{Eq}\left(g_{k},\varepsilon\right)$
for each $s=\max,\min$ and each $k$, and, if $k\in A\left(g,\varepsilon\right)$,
then $\left|\psi_{s}\left(k\right)-x_{s}\right|\leq\frac{1}{8}\eta$.
Then, for each subset $B\subseteq\left\{ 1,...,K\right\} $ of copies,
there is an equilibrium $a$ with average payoffs equal to 
\[
\text{Av}\left(a\right)=\frac{1}{K}\left(\sum_{k\in B}\psi_{\max}\left(k\right)+\sum_{k\notin B}\psi_{\min}\left(k\right)\right).
\]
Because of the choice of $\psi_{.}\left(.\right)$, 
\[
\frac{\left|B\right|}{K}x_{\max}+\frac{K-\left|B\right|}{K}x_{\min}-\frac{1}{2}\eta\leq\text{Av}\left(x\right)\leq\frac{\left|B\right|}{K}x_{\max}+\frac{K-\left|B\right|}{K}x_{\min}+\frac{1}{2}\eta.
\]
If $K\geq\frac{2}{\eta}$, for any $x\in\left[x_{\min},x_{\max}\right]$,
we can choose $B$, and hence arrive at equilibrium $a$, so that
the average payoffs in $a$ are at most $\eta$-far from $x$, $\left|\text{Av}\left(a\right)-x\right|\leq\eta$. 

\section{Proof of Theorem \ref{thm:Largest equilibrium}\label{sec:Proof-Largest Eq}}

The first subsection introduces notation and metric $d$. Section
\ref{subsec:Deterministic-bounds} derives various deterministic inequalities
connecting metric $d$ and average behavior. Section \ref{subsec:Bounds-on-a}
derives probabilistic bounds. The next two sections contain steps
(a) and (b) described in the introduction. The last section concludes
the proof of the theorem. 

We begin with preliminary remarks. It is enough to establish one side
of the probability bound: for each $\eta>0$ and $w<\infty$, there
is $\delta>0$, such that for each network $g$, if $d\left(g\right)\leq\delta$,
$w\left(g\right)\leq w$, then
\[
\Prob\text{\ensuremath{\left(\max\text{Eq}\left(g^{N},\varepsilon\right)\geq x_{\max}+\eta\right)}	\ensuremath{\leq\eta}.}
\]
 The proof of the other probability bound is analogous and the two
bounds together combine to the statement of the theorem. 

Say that $a$ is an upper equilibrium if, whenever indifferent, each
agent plays action $1$. Because of supermodularity, if $a$ is an
equilibrium, there exists $a^{\prime}\geq a$ that is an upper equilibrium.
Thus, it is enough to show the above probability bound when set $\text{Eq}\left(g,\varepsilon\right)$
contains only the average payoffs in all upper equilibria. 

Because $x_{\max}$ is strongly stable, there exists a constant $\gamma<1$
such that for each $x$,
\[
P\left(x\right)\leq\max\left(x_{\max},x_{\max}+\left(1-\gamma\right)\left(x-x_{\max}\right)\right)=P^{*}\left(x\right).
\]
(Such constant exists locally due to the definition of strong stability.
The existence for all $x$ follows from compactness and the fact that
$x_{\max}$ is the largest fixed point of $P$.) Because $P^{*}$
is increasing and right-continuous, Lemma \ref{lem:Dominance} implies
that there exists a random utility game $\left(u^{*},F^{*}\right)$
with continuum best response function $P^{*}$ that dominates $\left(u,F\right)$.
In particular, it is enough to show the second claim in Theorem \ref{thm:Largest equilibrium}
for game $\left(u^{*},F^{*}\right)$. Henceforth, we assume that $P^{*}$
is the continuum best response function. Notice that $P^{*}$ is Lipschitz
with a Lipschitz constant equal to $\gamma$.

\subsection{Notation}

For each profile $a$, let $\beta_{i}^{a}=\frac{1}{g_{i}}\sum g_{ij}a_{j}$
be the average behavior of neighbors of $i$. Let $b_{i}\left(a,\varepsilon\right)=\max\left(\arg\max_{a_{i}}u_{i}\left(a_{i},\beta_{i}^{a},\varepsilon\right)\right)$
be the largest best response action of agent $i$ against $a_{-i}$
given payoff shock $\varepsilon_{i}$. Let $b\left(a,\varepsilon\right)$
be the profile of best responses. If $a$ is an upper equilibrium
given $\varepsilon$, then $b\left(a,\varepsilon\right)=a$. Also,
we denote $p^{a}=\left(P^{*}\left(\beta_{i}^{a}\right)\right)_{i}$
to be the profile of expected best responses. 

Let $\mathcal{A}=\left[0,1\right]^{N}$ be the space of (mixed) action
profiles. Let 
\[
\mathcal{B}=\left\{ \beta^{a}:a\in\mathcal{A}\right\} 
\]
be the set of profiles $\beta^{a}=\left(\beta_{i}^{a}\right)$ of
neighborhood behaviors that can be generated from the profiles. We
assume that $\mathcal{A}$ is a subset of a normed space $\R^{N}$
with a norm-induced metric
\[
d\left(a,b\right)=\sqrt{\frac{1}{\sum g_{i}^{2}}\sum g_{i}^{2}\left(a_{i}-b_{i}\right)^{2}}.
\]
 This is a weighted Euclidean metric normalized so that the diameter
of $\mathcal{A}$ for a balanced graph is equal to $\text{diam}\mathcal{A}=1$. 

Let $g_{\min}=\min_{i}g_{i}$ and $g_{\max}=\max_{i}g_{i}$. 

\subsection{Deterministic relationships\label{subsec:Deterministic-bounds}}
\begin{lem}
\label{lem:Averages}For each profile $a\in\mathcal{A}$,
\[
\text{Av}\left(a\right)=\text{Av}\left(\beta^{a}\right).
\]
\end{lem}
\begin{proof}
Notice that 
\[
\text{Av}\left(\beta^{a}\right)=\frac{1}{\sum_{i}g_{i}}\sum_{i}g_{i}\frac{1}{g_{i}}\sum_{j}g_{ij}a_{j}=\frac{1}{\sum_{i}g_{i}}\sum_{i}\sum_{j}g_{ij}a_{j}=\frac{1}{\sum_{i}g_{i}}\sum_{j}a_{j}g_{j}=\text{Av}\left(a\right).
\]
\end{proof}
\begin{lem}
\label{lem:Lipschitz averages}For any profiles $a,b\in\mathcal{A},$
\begin{align*}
 & \left|\text{Av}\left(P^{*}\left(a\right)\right)-\text{Av}\left(P^{*}\left(b\right)\right)\right|\leq\left|\text{Av}\left(a\right)-\text{Av}\left(b\right)\right|.
\end{align*}
\end{lem}
\begin{proof}
The inequality follows from $P^{*}$ being Lipschitz with a constant
$\gamma<1$. 
\end{proof}
\begin{lem}
\label{lem:avergaes vs metric}For any profiles $a,b\in\mathcal{A},$
\begin{align*}
\left|\text{Av}\left(a\right)-\text{Av}\left(b\right)\right| & \leq\sqrt{w\left(g\right)}d\left(a,b\right).
\end{align*}
\end{lem}
\begin{proof}
Notice that
\begin{align*}
\left|\text{Av}\left(a\right)-\text{Av}\left(b\right)\right| & \leq\frac{1}{\sum g_{i}}\sum g_{i}\left|a_{i}-b_{i}\right|\leq\sqrt{\frac{1}{\sum g_{i}}\sum g_{i}\left(a_{i}-b_{i}\right)^{2}}\\
 & \leq\sqrt{w\left(g\right)\frac{1}{\sum g_{i}^{2}}\sum g_{i}^{2}\left(a_{i}-b_{i}\right)^{2}}=\sqrt{w\left(g\right)}d\left(a,b\right),
\end{align*}
where the second inequality follows from the Jensen's inequality,
and the third one from $\sum g_{j}^{2}\leq g_{\max}\sum g_{j}\leq w\left(g\right)g_{i}\sum g_{j}$
for each $i$. 
\end{proof}
\begin{lem}
\label{lem: how much of a profile below xmax}Suppose that profile
$b$ is such that  $b_{i}\geq x_{\max}$ for each $i$. Then, for
each profile $a$,
\[
\frac{1}{\sum_{i}g_{i}}\sum_{i}g_{i}\max\left(x_{\max}-a_{i},0\right)\leq\sqrt{w\left(g\right)}d\left(a,b\right)
\]
\end{lem}
\begin{proof}
For each profile $a$, define profile $\min\left(x_{\max},a\right)$
so that $\left(\text{min}\left(x_{\max},a\right)\right)_{i}=\text{min}\left(x_{\max},a_{i}\right)$.
Then, because function $f\left(y\right)=\min\left(y,x_{\max}\right)$
is Lipschitz with constant 1, we have 
\[
d\left(\min\left(a,x_{\max}\right),x_{\max}\right)=d\left(\min\left(a,x_{\max}\right),\min\left(b,x_{\max}\right)\right)\leq d\left(a,b\right),
\]
where, abusing notation, we write $x_{\max}$ to denote the constant
profile, and we use the fact that $\min\left(b,x_{\max}\right)=x_{\max}$.
By Lemma \ref{lem:avergaes vs metric}, 
\begin{align*}
\frac{1}{\sum_{i}g_{i}}\sum_{i}g_{i}\max\left(x_{\max}-a_{i},0\right) & =\text{Av}\left(x_{\max}\right)-\text{Av}\left(\min\left(a,x_{\max}\right)\right)\\
 & =\text{Av}\left(\min\left(b,x_{\max}\right)\right)-\text{Av}\left(\min\left(a,x_{\max}\right)\right)\\
 & \leq\sqrt{w\left(g\right)}d\left(\min\left(a,x_{\max}\right),\min\left(b,x_{\max}\right)\right)\leq\sqrt{w\left(g\right)}d\left(a,b\right).
\end{align*}
\end{proof}
\begin{lem}
\label{lem:lower best responses}Suppose that profile $b$ is such
that $b_{i}\geq x_{\max}$ for each $i$. Then, for each profile $a$,
\[
\text{Av}\left(a\right)-\text{Av}\left(P^{*}\left(a\right)\right)\geq\left(1-\gamma\right)\left(\text{Av}\left(a\right)-x_{\max}\right)-2\left(w\left(g\right)\right)^{\frac{1}{4}},
\]
where $P^{*}\left(a\right)$ is a profile of actions $P^{*}\left(a_{i}\right)$
for each agent $i$.
\end{lem}
\begin{proof}
First, Lemma \ref{lem: how much of a profile below xmax} implies
that 

\[
\frac{1}{\sum_{i}g_{i}}\sum_{i}g_{i}\max\left(x_{\max}-a_{i},0\right)\leq\sqrt{w\left(g\right)}d\left(a,b\right)=\delta.
\]
Second, let $A=\sum_{i:a_{i}\leq x_{\max}-\sqrt{\delta}}g_{i}$ and
notice that 
\begin{align*}
A & =\sum_{i:a_{i}\leq x_{\max}-\sqrt{\delta}}g_{i}\leq\frac{1}{\sqrt{\delta}}\sum_{i:a_{i}\leq x_{\max}-\sqrt{\delta}}g_{i}\left(x_{\max}-a_{i}\right)\leq\frac{1}{\sqrt{\delta}}\sum_{i}g_{i}\max\left(x_{\max}-a_{i},0\right)\\
 & \leq\frac{\delta}{\sqrt{\delta}}\sum g_{i}=\sqrt{\delta}\sum g_{i}.
\end{align*}
Hence
\begin{align*}
 & \text{Av}\left(a\right)-\text{Av}\left(P^{*}\left(a\right)\right)=\frac{1}{\sum g_{i}}\sum_{i}g_{i}\left(a_{i}-P^{*}\left(a_{i}\right)\right)\\
\geq & -\frac{1}{\sum g_{i}}\sum_{i:a_{i}\leq x_{\max}-\sqrt{\delta}}g_{i}-\sqrt{\delta}\frac{1}{\sum g_{i}}\sum_{i:x_{\max}\geq a_{i}\geq x_{\max}-\sqrt{\delta}}g_{i}+\frac{1}{\sum g_{i}}\sum_{i:a_{i}\geq x_{\max}}g_{i}\left(a_{i}-P^{*}\left(a_{i}\right)\right)\\
\geq & -\frac{1}{\sum g_{i}}A-\sqrt{\delta}+\left(1-\gamma\right)\frac{1}{\sum g_{i}}\sum_{i:a_{i}\geq x_{\max}}g_{i}\left(a_{i}-x_{\max}\right)\\
\geq & \left(1-\gamma\right)\left(\frac{1}{\sum g_{i}}\sum_{i}g_{i}\left(a_{i}-x_{\max}\right)\right)-2\sqrt{\delta}=\left(1-\gamma\right)\left(\text{Av}\left(a\right)-x_{\max}\right)-2\sqrt{\delta}.
\end{align*}
\end{proof}

\subsection{Bounds on a probability that a profile is an equilibrium\label{subsec:Bounds-on-a}}

This subsection contains probabilistic bounds on the distances between
profiles of neighborhood behaviors. First, we show that the distance
between neighborhood behaviors obtained from the best response and
the expected best response profiles are small. Recall that, for any
profile $a$, $p^{a}$ is a profile of expected best responses: $p_{i}^{a}=P^{*}\left(\beta_{i}^{a}\right)$.
\begin{lem}
\label{lem:Prob bound 1}There exists a universal constant $c<\infty$
such that, for each profile $a$, 
\begin{align*}
 & \Prob\left(d\left(\beta^{b\left(a,\varepsilon\right)},\beta^{p^{a}}\right)\geq\eta\right)\leq\exp\left(-\frac{c}{\left(w\left(g\right)\right)^{4}}N\left(\eta^{2}-d\left(g\right)\right)^{2}\right).
\end{align*}
\end{lem}
\begin{proof}
Notice that 
\begin{align*}
 & \left(d\left(\beta^{b\left(a,\varepsilon\right)},\beta^{p^{a}}\right)\right)^{2}=\sum g_{i}^{2}\left(\beta_{i}^{b\left(a,\varepsilon\right)}-\beta_{i}^{p^{a}}\right)^{2}\\
= & \sum_{i}\left(\sum_{j}g_{ij}\left(b_{j}\left(a,\varepsilon\right)-p_{j}^{a}\right)\right)^{2}\\
= & \sum_{j\neq k}\left(\sum_{i}g_{ji}g_{ik}\right)\left(b_{j}\left(a,\varepsilon\right)-p_{j}^{a}\right)\left(b_{k}\left(a,\varepsilon\right)-p_{k}^{a}\right)+\sum_{j}\left(\sum_{i}g_{ij}^{2}\right)\left(b_{j}\left(a,\varepsilon\right)-p_{j}^{a}\right)^{2}.
\end{align*}
Because $g_{ij}\leq d\left(g\right)g_{i}$, the second term is not
larger than $d\left(g\right)\sum g_{i}^{2}.$ Let $x_{j}=b_{j}\left(a,\varepsilon\right)-p_{j}^{a}$
for each $j$. Then,
\begin{align*}
\Prob\left(d\left(\beta^{b\left(a,\varepsilon\right)},\beta^{p^{a}}\right)\geq\eta\right)\leq & \Prob\left(\sum_{j\neq k}\left(\sum_{i}g_{ji}g_{ik}\right)x_{j}x_{k}\geq\left(\eta^{2}-d\left(g\right)\right)\sum_{i}g_{i}^{2}\right).
\end{align*}

Let $g_{jk}^{\left(2\right)}=\sum_{i}g_{ji}g_{ik}$ and let $G^{\left(2\right)}$
be the symmetric matrix of elements $g_{jk}^{\left(2\right)}$. Observe
that 
\[
g_{jk}^{\left(2\right)}=\sum_{i}g_{ji}g_{ik}=\sum_{i}\frac{g_{ji}g_{ik}}{g_{j}g_{i}}g_{j}g_{i}\leq\left(w\left(g\right)\right)^{2}g_{\min}^{2}\pi_{jk},
\]
where we denote $\pi_{jk}=\dot{\sum_{i}\frac{g_{ji}}{g_{j}}\frac{g_{ik}}{g_{i}}}$.
Note that, for each $j$, $\sum_{k}\pi_{jk}=\sum_{k,i}\frac{g_{ji}}{g_{j}}\frac{g_{ik}}{g_{i}}=\sum_{i}\frac{g_{ji}}{g_{j}}=1$.
Hence $\pi_{jk}\leq1$.

Because the best response of each player $i$ depends only on independent
shock $\varepsilon_{i}$ (and not on other payoff shocks), $x_{j}$
and $x_{k}$ are independent for $j\neq k$. Hence the expected value
of $\sum_{j\neq k}\left(\sum_{i}g_{ji}g_{ik}\right)x_{j}x_{k}$ is
equal to 0, and we can use the Hansen-Wright inequality (Theorem 6.2.1
in \cite{vershynin_high-dimensional_2018}): 
\[
\Prob\left(\sum_{j\neq k}\left(\sum_{i}g_{ji}g_{ik}\right)x_{j}x_{k}\geq t\right)\leq2\exp\left(-ct^{2}\left\Vert G^{\left(2\right)}\right\Vert _{F}^{-2}\right),
\]
where $c$ is some universal constant (note that the random variables
$x_{j}$ are bound by 2), and where $\left\Vert G^{\left(2\right)}\right\Vert _{F}$
is the Frobenius norm of matrix $G^{\left(2\right)}$:
\begin{align*}
\left\Vert G^{\left(2\right)}\right\Vert _{F}^{2} & =\sum_{i}\sum_{j}\left(g_{ij}^{\left(2\right)}\right)^{2}\leq\left(w\left(g\right)\right)^{4}g_{\min}^{4}\sum_{i}\sum_{j}\pi_{ij}^{2}\\
 & \leq\left(w\left(g\right)\right)^{4}g_{\min}^{4}\sum_{i}\sum_{j}\pi_{jk}\leq\left(w\left(g\right)\right)^{4}g_{\min}^{4}N.
\end{align*}
Take $t=\left(\eta^{2}-d\left(g\right)\right)\sum_{i}g_{i}^{2}$,
and notice that $\sum_{i}g_{i}^{2}\geq Ng_{\min}^{2}$ to obtain the
inequality in the statement of the lemma. 
\end{proof}
The second result shows that, for any fixed profile $a_{0}$, the
maximum distance between neighborhood behaviors obtained as the best
response to $a_{0}$ and the best response to some other profile $a$,
across all profiles $a$ that have similar neighborhood behaviors
to $a_{0}$, is small. 
\begin{lem}
\label{lem:Prob bound 2}For each profile $a_{0}$,
\[
\Prob\left(\sup_{a:d\left(\beta^{a},\beta^{a_{0}}\right)\leq\delta}d\left(\beta^{b\left(a_{0},\varepsilon\right)},\beta^{b\left(a,\varepsilon\right)}\right)\geq\eta\right)\leq\exp\left(-\frac{1}{2\left(w\left(g\right)\right)^{4}}N\left(\eta-3\delta^{2/3}\right)^{2}\right).
\]
\end{lem}
\begin{proof}
For each profile $a$ and player $i$, $b_{i}\left(a,\varepsilon\right)\neq b\left(a_{0},\varepsilon\right)$
if and only if either $\beta_{i}^{a}\leq\beta\left(\varepsilon\right)<\beta_{i}^{a_{0}}$
or $\beta_{i}^{a_{0}}\leq\beta\left(\varepsilon\right)<\beta_{i}^{a}$.
Denote a random variable 
\[
X_{i}=\mathbf{1}\left\{ \beta\left(\varepsilon_{i}\right)\in\left[\beta_{i}^{a_{0}}-\delta^{2/3},\beta_{i}^{a_{0}}+\delta^{2/3}\right]\right\} .
\]
Then, for any profile $a$, 
\[
\left|b_{i}\left(a,\varepsilon\right)-b\left(a_{0},\varepsilon\right)\right|\leq X_{i}\mathbf{1}\left\{ \left|\beta_{i}^{a}-\beta_{i}^{a_{o}}\right|\leq\delta^{2/3}\right\} +\mathbf{1}\left\{ \left|\beta_{i}^{a}-\beta_{i}^{a_{o}}\right|>\delta^{2/3}\right\} ,
\]
and 
\begin{align*}
 & \sup_{a:d\left(\beta^{a},\beta^{a_{0}}\right)\leq\delta}\sum_{i}g_{i}^{2}\left(b_{i}\left(a,\varepsilon\right)-b_{i}\left(a_{0},\varepsilon\right)\right)^{2}\\
\leq & \sum g_{i}^{2}X_{i}^{2}+\sup_{a:d\left(\beta^{a},\beta^{a_{0}}\right)\leq\delta}\sum_{i:\left|\beta_{i}^{a}-\beta_{i}^{a_{o}}\right|>\delta^{2/3}}g_{i}^{2}\\
\leq & \sum g_{i}^{2}X_{i}^{2}+\sup_{a:d\left(\beta^{a},\beta^{a_{0}}\right)\leq\delta}\delta^{-4/3}\sum_{i:\left|\beta_{i}^{a}-\beta_{i}^{a_{o}}\right|>\delta^{2/3}}g_{i}^{2}\left(\beta_{i}^{a}-\beta_{i}^{a_{o}}\right)^{2}\\
\leq & \sum g_{i}^{2}X_{i}^{2}+\delta^{-4/3}\sup_{a:d\left(\beta^{a},\beta^{a_{0}}\right)\leq\delta}\sum g_{i}^{2}\left(d\left(\beta^{a},\beta^{a_{0}}\right)\right)^{2},\\
\leq & \sum g_{i}^{2}X_{i}^{2}+\delta^{2}\delta^{-4/3}\sum g_{i}^{2}=\sum g_{i}^{2}X_{i}^{2}+\delta^{2/3}\sum g_{i}^{2}.
\end{align*}

Variables $X_{i}^{2}=X_{i}$ are independent Bernoulli variables with
parameter $\E X_{i}=P^{*}\left(\beta_{i}^{a}+\delta^{2/3}\right)-P^{*}\left(\beta_{i}^{a}+\delta^{2/3}\right)\leq2\delta^{2/3}$
as $P^{*}$ is Lipschitz with constant 1. The Hoeffding's inequality
shows that 
\begin{align*}
\Prob\left(\sum g_{i}^{2}X_{i}^{2}+\delta^{2/3}\sum g_{i}^{2}\geq\eta\sum g_{i}^{2}\right) & \leq\Prob\left(\sum g_{i}^{2}\left(X_{i}-\E X_{i}\right)\geq\left(\eta-3\delta^{2/3}\right)\sum g_{i}^{2}\right)\\
 & \leq\exp\left(-\frac{\left(\sum g_{i}^{2}\right)^{2}}{2\sum g_{i}^{4}}\left(\eta-3\delta^{2/3}\right)^{2}\right).
\end{align*}
Finally, notice that $2\sum g_{i}^{4}\leq\left(w\left(g\right)\right)^{4}g_{\min}^{4}N$
and $\left(\sum g_{i}^{2}\right)^{2}\geq g_{\min}^{4}N^{2}$.
\end{proof}

\subsection{Probability bound on the local existence of an upper equilibrium\label{subsec:Probability-bound-on}}

This subsection finds a bound on the probability that, for any profile
$a_{0}$, there exists a profile $a$ with similar neighborhood behaviors
as $a_{0}$, and such that~$a$ is an upper equilibrium. 
\begin{lem}
\label{lem:Uper equilibrium}For each $\xi>0$ and each $w<\infty$,
there is $\delta>0$ so that, for each profile $a_{0}$ such that
$\text{Av}\left(a_{0}\right)>x_{\max}+\xi$, and for each network
$g$ such that $d\left(g\right)\leq\delta$ and $w\left(g\right)\leq w$,
\[
\Prob\left(\text{there exists }a\text{ s.t. }a\text{ is upper equilibrium and }d\left(\beta^{a},\beta^{a_{0}}\right)\leq\delta\right)\leq2\exp\left(-\delta N\right).
\]
\end{lem}
\begin{proof}
Choose $\eta,\delta>0$ such that 
\begin{align*}
\left(1-\gamma\right)\xi & >\sqrt{w\left(g\right)}\left(2\delta+2\eta\right)+2\left(w\left(g\right)\right)^{\frac{1}{4}}\sqrt{2\eta}\text{ and }\\
\delta & \leq\frac{c}{\left(w\left(g\right)\right)^{4}}\left(\eta^{2}-\delta\right)^{2}+\frac{1}{2\left(w\left(g\right)\right)^{4}}\left(\eta-3\delta^{2/3}\right)^{2}.
\end{align*}
Assume that $d\left(g\right)\leq\delta$. 

Consider the following three events: 
\begin{align*}
A & =\left\{ d\left(\beta^{b\left(a_{0},\varepsilon\right)},\beta^{p^{a_{0}}}\right)\leq\eta\right\} ,\\
B & =\left\{ \sup_{a:d\left(\beta^{a},\beta^{a_{0}}\right)\leq\delta}d\left(\beta^{b\left(a_{0},\varepsilon\right)},\beta^{b\left(a,\varepsilon\right)}\right)\leq\eta\right\} .
\end{align*}
Due to Lemmas \ref{lem:Prob bound 1} and \ref{lem:Prob bound 2},
the probability that at least one of the two events does not hold
is no larger than 
\begin{align*}
 & \exp\left(-\frac{c}{\left(w\left(g\right)\right)^{4}}N\left(\eta^{2}-d\left(g\right)\right)^{2}\right)+\exp\left(-\frac{1}{2\left(w\left(g\right)\right)^{4}}N\left(\eta-3\delta^{2/3}\right)^{2}\right)\leq2\exp\left(-\delta N\right).
\end{align*}

Assume that the two events hold simultaneously. We will show that
there exists no $a$ such that $d\left(\beta^{a},\beta^{a_{0}}\right)\leq\delta$
and such that $a$ is an upper equilibrium. 

On the contrary, suppose that such $a$ exists. Then, $a=b\left(a,\varepsilon\right)$.
Because $d$ is a metric and events $A$ and $B$ hold, 
\[
d\left(\beta^{a},\beta^{p^{a_{0}}}\right)=d\left(\beta^{b\left(a,\varepsilon\right)},\beta^{p^{a_{0}}}\right)\leq d\left(\beta^{b\left(a,\varepsilon\right)},\beta^{b\left(a_{0},\varepsilon\right)}\right)+d\left(\beta^{b\left(a_{0},\varepsilon\right)},\beta^{p^{a_{0}}}\right)\leq2\eta.
\]
Because $\beta_{i}^{p^{a_{o}}}=\frac{1}{g_{i}}\sum_{i}g_{ij}P^{*}\left(a_{0,j}\right)\geq x_{\max}$
for each $i$, we can apply Lemma \ref{lem:lower best responses}
to $\beta^{a}$ instead of $a$ and $\beta^{p^{a_{0}}}$ instead of
$b$ (notice that $p^{a}=P^{*}\left(\beta^{a}\right)$ by definition):
\begin{equation}
\text{Av}\left(\beta^{a}\right)-\text{Av}\left(p^{a}\right)\geq\left(1-\gamma\right)\left(\text{Av}\left(\beta^{a}\right)-x_{\max}\right)-2\left(w\left(g\right)\right)^{\frac{1}{4}}\sqrt{2\eta}.\label{eq:ineq 1}
\end{equation}

By Lemmas \ref{lem:Averages}, \ref{lem:Lipschitz averages}, and
\ref{lem:avergaes vs metric}, and because $d\left(\beta^{a},\beta^{a_{0}}\right)\leq\delta$,
\begin{align*}
\left|\text{Av}\left(p^{a}\right)-\text{Av}\left(p^{a_{0}}\right)\right| & \leq\left|\text{Av}\left(a\right)-\text{Av}\left(a_{0}\right)\right|=\left|\text{Av}\left(\beta^{a}\right)-\text{Av}\left(\beta^{a_{0}}\right)\right|\leq\sqrt{w\left(g\right)}\delta.
\end{align*}
By Lemmas \ref{lem:Averages} and \ref{lem:avergaes vs metric}, and
because event $A$ holds,
\begin{align*}
\left|\text{Av}\left(p^{a_{0}}\right)-\text{Av}\left(b\left(a_{0},\varepsilon\right)\right)\right| & =\left|\text{Av}\left(\beta^{p^{a_{0}}}\right)-\text{Av}\left(\beta^{b\left(a_{0},\varepsilon\right)}\right)\right|\leq\sqrt{w\left(g\right)}\eta.
\end{align*}
By Lemmas \ref{lem:Averages} and \ref{lem:avergaes vs metric}, because
$a$ is an upper equilibrium, and because event $B$ holds,
\begin{align*}
\left|\text{Av}\left(b\left(a_{0},\varepsilon\right)\right)-\text{Av}\left(\beta^{a}\right)\right| & =\left|\text{Av}\left(b\left(a_{0},\varepsilon\right)\right)-\text{Av}\left(a\right)\right|\\
 & =\left|\text{Av}\left(b\left(a_{0},\varepsilon\right)\right)-\text{Av}\left(b\left(a,\varepsilon\right)\right)\right|\\
 & =\left|\text{Av}\left(\beta^{b\left(a_{0},\varepsilon\right)}\right)-\text{Av}\left(\beta^{b\left(a,\varepsilon\right)}\right)\right|\leq\sqrt{w\left(g\right)}\eta.
\end{align*}
Putting the three inequalities together, we obtain
\begin{align}
 & \left|\text{Av}\left(\beta^{a}\right)-\text{Av}\left(p^{a}\right)\right|\nonumber \\
 & \leq\left|\text{Av}\left(b\left(a_{0},\varepsilon\right)\right)-\text{Av}\left(\beta^{a}\right)\right|+\left|\text{Av}\left(b\left(a_{0},\varepsilon\right)\right)-\text{Av}\left(p^{a_{a}}\right)\right|+\left|\text{Av}\left(p^{a}\right)-\text{Av}\left(p^{a_{0}}\right)\right|\nonumber \\
 & \leq\sqrt{w\left(g\right)}\left(\delta+2\eta\right).\label{eq:ineq 2}
\end{align}

Combining inequalities (\ref{eq:ineq 1}) and (\ref{eq:ineq 2}),
we obtain
\[
\left(\sqrt{w\left(g\right)}\left(\delta+2\eta\right)+2\left(w\left(g\right)\right)^{\frac{1}{4}}\sqrt{2\eta}\right)\geq\left(1-\gamma\right)\left(\text{Av}\left(\beta^{a}\right)-x_{\max}\right).
\]
By Lemmas \ref{lem:Averages}, \ref{lem:avergaes vs metric}, and
because $d\left(\beta^{a},\beta^{a_{0}}\right)\leq\delta$,
\begin{align*}
\left|\text{Av}\left(\beta^{a}\right)-\text{Av}\left(a_{0}\right)\right| & =\left|\text{Av}\left(\beta^{a}\right)-\text{Av}\left(\beta^{a_{0}}\right)\right|\leq\sqrt{w\left(g\right)}\delta.
\end{align*}
Hence, 
\[
\left(\sqrt{w\left(g\right)}\left(2\delta+2\eta\right)+2\left(w\left(g\right)\right)^{\frac{1}{4}}\sqrt{2\eta}\right)\geq\left(1-\gamma\right)\left(\text{Av}\left(a_{0}\right)-x_{\max}\right)\geq\left(1-\gamma\right)\xi.
\]
However, this violates the choice of the parameters $\eta$ and $\delta$.
\end{proof}

\subsection{Metric entropy bound\label{subsec:Metric-entropy-bound}}

For each $\delta>0$, let $\mathcal{N}\left(\delta,\mathcal{B}\right)$
be the covering number of $\mathcal{B}$, i.e., the smallest cardinality
$n$ of a list of profiles $b^{1},...,b^{n}\in\mathcal{B}$ such that
for each $b\in\mathcal{B}$, there is $l\leq n$ so that $d\left(b,b^{l}\right)\leq\delta$. 
\begin{lem}
\label{lem:Sudakov}There exists a constant $c<\infty$ such that,
for each $\delta>0$, and each network $g$, 
\[
\mathcal{N}\left(\delta,\mathcal{B}\right)\leq\exp\left(\frac{1}{\delta^{2}}c\left(w\left(g\right)\right)^{2}d\left(g\right)N\right).
\]
. 
\end{lem}
\begin{proof}
We will use the Sudakov's Minoration Inequality (Theorem 7.4.1 from
\cite{vershynin_high-dimensional_2018}) which provides an upper bound
on the covering number via the expectation of a certain Gaussian process.
For this, let $Z_{i}$ for each agent $i$ be an i.i.d. standard normal
random variable. For each (possibly mixed) profile $a\in\mathcal{A}$,
define 
\[
X_{a}=\frac{1}{\sqrt{\sum_{i}g_{i}^{2}}}\sum_{i}g_{i}a_{i}Z_{i}.
\]
For any two profiles $a,b\in\mathcal{A}$,
\begin{align*}
\sqrt{\E\left(X_{a}-X_{b}\right)^{2}} & =\sqrt{\frac{1}{\sqrt{\sum g_{i}^{2}}}\E\left(\sum_{i}g_{i}\left(a_{i}-b_{i}\right)Z_{i}\right)^{2}}\\
 & =\sqrt{\frac{1}{\sqrt{\sum g_{i}^{2}}}\sum_{i}g_{i}\left(a_{i}-b_{i}\right)^{2}}=d\left(a,b\right).
\end{align*}

The Sudakov's Minoration Inequality implies that, for some universal
(i.e., independent of parameters and a current problem) constant $c_{1}>0$,
\[
\log\mathcal{N}\left(\delta,\mathcal{B}\right)\leq c_{1}\frac{\left(\E\sup_{b\in\mathcal{B}}X_{b}\right)^{2}}{\delta^{2}}.
\]
We compute 
\begin{align*}
\E\sup_{b\in\mathcal{B}}X_{b} & =\E\sup_{a\in\mathcal{A}}X_{\beta^{a}}=\E\left(\sup_{a\in\mathcal{A}}\frac{1}{\sqrt{\sum_{i}g_{i}^{2}}}\sum_{i}g_{i}Z_{i}\left(\frac{1}{g_{i}}\sum g_{ij}a_{j}\right)\right)\\
 & =\frac{1}{\sqrt{\sum_{i}g_{i}^{2}}}\E\left(\sup_{a\in\mathcal{A}}\sum_{i}a_{i}\left(\sum_{j}g_{ij}Z_{j}\right)\right)\leq\frac{1}{\sqrt{\sum_{i}g_{i}^{2}}}\E\sum_{i}\left|\sum_{j}g_{ij}Z_{j}\right|\\
 & \leq\sqrt{\frac{2}{\pi}}\frac{1}{\sqrt{\sum_{i}g_{i}^{2}}}\sum_{i}\sqrt{\sum_{j}g_{ij}^{2}},
\end{align*}
where the last inequality is due to a bound on the expectation of
the absolute value of the normal variable $\sum g_{ij}Z_{j}$ via
its standard deviation $\sigma_{i}=\sqrt{\sum_{j}g_{ij}^{2}}$. Because
$\sum_{j}g_{ij}^{2}\leq d\left(g\right)g_{i}^{2}$ and $\left(\sum_{i}g_{i}\right)^{2}\leq N^{2}\left(w\left(g\right)\right)^{2}g_{\min}^{2}\leq N\left(w\left(g\right)\right)^{2}\sum g_{i}^{2}$,
we have
\[
\log\mathcal{N}\left(\delta,\mathcal{B}\right)\leq\sqrt{\frac{2}{\pi}}c_{1}\frac{1}{\delta^{2}}\frac{1}{\sum_{i}g_{i}^{2}}\left(\sum_{i}\sqrt{d\left(g\right)}g_{i}\right)^{2}d\left(g\right)\leq\frac{1}{\delta^{2}}\sqrt{\frac{2}{\pi}}c_{1}\left(w\left(g\right)\right)^{2}d\left(g\right)N.
\]
\end{proof}

\subsection{Proof of Theorem \ref{thm:Largest equilibrium}}

Fix $\eta>0$ and $w<\infty$. Use Lemma \ref{lem:Uper equilibrium}
to find $\delta>0$ and $\delta\leq\frac{1}{2\sqrt{w}}\eta$ such
that, for each profile $b$, and each network $g$, if $\text{Av}\left(b\right)\geq x_{\max}+\frac{1}{2}\eta$,
$d\left(g\right)\leq\delta$, and $w\left(g\right)\leq w$, then 
\[
\Prob\left(\text{there exists }a\text{ s.t. }a\text{ is upper equilibrium and }d\left(\beta^{a},\beta^{b}\right)\leq\delta\right)\leq2\exp\left(-\delta N\right).
\]

Use Lemma \ref{lem:Sudakov} to find a list of $n\leq\exp\left(\frac{1}{\delta^{2}}c\left(w\left(g\right)\right)^{2}d\left(g\right)N\right)$
profiles $b^{1},...,b^{n}$ such that, for each profile $a\in\mathcal{A}$,
there is $l\leq n$ such that $d\left(\beta^{b^{l}},\beta^{a}\right)\leq\delta$.
Observe that if $a$ is such that $\text{Av}\left(a\right)>x_{\max}+\eta$
and $d\left(\beta^{b^{l}},\beta^{a}\right)\leq\delta$ for some $l$,
then, by Lemmas \ref{lem:Averages} and \ref{lem:avergaes vs metric},
\begin{align*}
\text{Av}\left(b^{l}\right)-\left(x_{\max}+\frac{1}{2}\eta\right) & \geq\text{Av}\left(a\right)-\left(x_{\max}+\eta\right)+\frac{1}{2}\eta-\left|\text{Av}\left(a\right)-\text{Av}\left(b^{l}\right)\right|\\
 & \geq\frac{1}{2}\eta-\left|\text{Av}\left(\beta^{a}\right)-\text{Av}\left(\beta^{b^{l}}\right)\right|\geq\frac{1}{2}\eta-\sqrt{w}d\left(\beta^{a},\beta^{b}\right)\geq0.
\end{align*}

Putting the above observations together yields
\begin{align*}
 & \Prob\left(\text{there exists }a\text{ st. }a\text{ is upper equilibrium}\text{ and }\text{Av}\left(a\right)\geq x_{\max}+\eta\right)\\
\leq & \sum_{l\leq n:\text{Av}\left(b^{l}\right)\geq x_{\max}+\frac{1}{2}\eta}\Prob\left(\text{there exists }a\text{ st. }a\text{ is upper equilibrium and }d\left(\beta^{a},\beta^{b^{l}}\right)\leq\delta\right)\\
\leq & 2\exp\left(-\left(\delta-\frac{1}{\delta^{2}}c\left(w\left(g\right)\right)^{2}d\left(g\right)\right)N\right)
\end{align*}
for some universal constant $c$. Because $N\geq\frac{1}{d\left(g\right)}$,
if 
\[
d\left(g\right)\leq\min\left(\frac{1}{2}\delta^{3}c^{-1}\left(w\left(g\right)\right)^{-2},\frac{1}{2}\frac{1}{\log2-\log\eta}\delta\log\eta\right),
\]
the above probability is smaller than $\eta$. 

\section{Proof of Theorem \ref{thm:RUdominant selection on lattice}\label{sec:Proof-of-Unique selection}}

\subsection{Proof description}

The proof is divided into five parts. Section \ref{subsec:Contagion-wave}
is devoted to the existence of a contagion wave, i.e., the third step
of the proof intuition from the main body of the paper.

Section \ref{subsec:Lattice} introduces a two-dimensional lattice.
In the limit, the neighborhoods converge to radius-1 balls in $\R^{2}$.

In Section \ref{subsec:Small-cubes}, we divide the lattice into areas,
called \emph{small cubes}, such that (a) there are many agents and
the law of large numbers can be applied to describe the empirical
distribution of payoff shocks inside each small cube, and (b) the
cubes are sufficiently small so that agents from the same small cube
have similar neighborhoods, which implies that their incentives are
similar. The two properties imply that average behavior in a small
cube is close to the behavior of a continuum of agents in the toy
model. 

Section \ref{lem:Good giant component} studies the statistical distribution
of bad small cubes, i.e. small cubes, where the empirical distribution
of payoff shocks is not close to the distribution from which the shocks
are drawn. We show that there are few of them and sufficiently sparse,
so that the set of small cubes which are far away from the bad cubes
contains a giant connected component. 

The last section concludes the proof of the theorem. 

\subsection{Contagion wave\label{subsec:Contagion-wave}}

Consider a toy model, where agents are located on a line, each location
has a continuum of agents, with a continuum best response function
$Q$ (not necessarily the same as $P$ from the statement of the theorem),
the connections depend only on the distance between agents, and and
the cumulative weight of connections between agents $x$ and agents
in set $\left\{ y^{\prime}:y^{\prime}\leq y\right\} $ is equal to
$f\left(y-x\right)$, where $f:\R\rightarrow\left[0,1\right]$ is
a function that is \emph{balanced}: (a) $f\left(x\right)$ is strictly
increasing for $x\in\left(-1,1\right)$, and (b) $f\left(-1\right)=0$
and $f\left(x\right)+f\left(-x\right)=1$ for each $x$. Given the
interpretation of $f$ stated above, condition $f\left(x\right)+f\left(-x\right)=1$
is a consequence of the symmetry of the connection weights, and $f\left(-1\right)=0$
means that agents separated by 1 or more are not connected. Notice
that the weight of connections depends only on the distance between
the agents. 

Consider a strategy $\sigma$ that is increasing in locations. For
each location $x,$ the average action of neighbors of agents in location
$x$ is equal to (assuming enough regularity, for intuition) 
\[
\int\sigma\left(y\right)df\left(y-x\right)=\lim_{a\rightarrow-\infty}\sigma\left(a\right)+\int\left(1-f\left(y-x\right)\right)d\sigma\left(y\right).
\]
We say that $\sigma$ is a contagion wave for $Q$ if, at each location
$x$, the best response of agents in such a location no higher than
$\sigma\left(x\right)$ or, in other words, if the above average action
is smaller than $Q^{-1}\left(\sigma\left(x\right)\right)$. 

This section contains two results: first, we show the existence of
a contagion wave for a continuum best response function that can be
represented by a step function, and next, we show the existence of
a stronger version of a wave for the original best response function
$P$. 

We begin with a definition. An increasing function $q:\R\rightarrow\left[0,1\right]$
is a step function if the image $q\left(\R\right)$ is finite. We
refer to the elements of the image as steps. If $q$ is a step function
and $a\in q\left(\R\right)$ is a step, then the most recent step
before $a$ is denoted as $a_{-}=\max\left\{ b\in q\left(\R\right):b<a\right\} $.
For each $a\in\left[0,1\right]$, let $q^{-1}\left(a\right)=\min\left(v:q\left(v\right)\geq a\right)$
if the set is non-empty and $q^{-1}\left(a\right)=\infty$ if the
set is empty. We have $q^{-1}\left(a_{-}\right)<q^{-1}\left(a\right)$
for each step $a$.
\begin{lem}
\label{lem:Wave function} Let $Q$ be a step function with steps
$0\leq a_{0}<...<a_{L+1}=1$ and such that for each $a>a_{0}$, we
have 
\begin{equation}
\intop_{a_{0}}^{a}\left(Q^{-1}\left(x\right)-x\right)dx>0.\label{eq:RU wave}
\end{equation}
Suppose that $f$ is a continuous and balanced function. Then, there
exist $0=v_{0}<v_{1}...<v_{L}\leq L$ such that, for each $l=1,...,L$,
\[
Q^{-1}\left(a_{l+1}\right)\geq a_{0}+\sum_{k\geq0}\left(1-f\left(v_{k}-v_{l}\right)\right)\left(a_{k+1}-a_{k}\right).
\]
\end{lem}
We interpret each vector as a step strategy, where agents in locations
$v_{l-1}<x\leq v_{l}$ play action $a_{l}$. Then, the right-hand
side of the inequality is equal to the average action experiences
in location $v_{l}$. The lemma says that, if $Q$ is a step function,
and it satisfies condition (\ref{eq:RU wave}), then we can choose
the step strategy such that the next step action $a_{l+1}$ is a ($Q$-)best
response for agents living on threshold $v_{l}$. 
\begin{proof}
Let $V$ be the set of all vectors $v=\left(v_{0},...,v_{L}\right)$
such that
\begin{align*}
0 & =v_{0}\leq...\leq v_{L}\text{ and }v_{l+1}\leq v_{l}+1\text{ fov each }l=0,...,L-1.
\end{align*}
 (Abusing notation, we take $v_{-1}=-\infty$.) Define function $F:\left[-1,L+1\right]\times V\rightarrow\R$
so that 
\[
F\left(x|v\right)=a_{0}+\sum_{k\geq0}\left(1-f\left(v_{k}-x\right)\right)\left(a_{k+1}-a_{k}\right).
\]
Then, for each strategy $v$, $F$ is the (weighted) average action
experienced by agents in location $x$. 

Due to properties of function $f$, function $F$ is continuous, strictly
increasing in $x$ for $x\in\left(v_{0}-1,v_{L-1}+1\right)$ and decreasing
in the lattice order on $V^{*}$ (i.e., $F\left(x,v\right)\geq F\left(x,v^{\prime}\right)$
for any $v,v^{\prime}$ such that $\forall_{k}v_{k}\leq v_{k}^{\prime}$.)
For each $v\in V$ and $l=1,...,L$, define
\[
b_{l}\left(v\right)=\inf\left\{ x\geq0:F\left(x|v\right)\geq Q^{-1}\left(a_{l+1}\right)\right\} ,
\]
and we take $b_{l}\left(v\right)=\infty$ if the set is empty. $b_{l}\left(v\right)$
is the first location in which action $a_{l+1}$ or higher is the
best response given the strategy determined by $v$. The properties
of $F$ imply that $b_{l}$ isweakly increasing in the lattice order
on $V$, and, because $Q^{-1}\left(a_{l+1}\right)>Q^{-1}\left(a_{l}\right)$,
we have $b_{l}\left(v\right)\leq b_{l+1}\left(v\right)$, with a strict
inequality if either $b_{l}\left(v\right)\in\left(0,\infty\right)$
or $b_{l+1}\left(v\right)\in\left(0,\infty\right)$. It is also continuous
for $v$ such that $b_{l}\left(v\right)<\infty$. Let $b\left(v\right)=\left(b_{l}\left(v\right)\right)_{l=1}^{L}$

Define function $b^{*}:V\rightarrow V$ so that 
\[
b_{l}^{*}\left(v\right)=\min\left(b_{l}\left(v\right),v_{l-1}+1\right),\text{ for each }l=1,...,L-1.
\]
Then, $b_{l}^{*}\left(v\right)\geq0$, $b^{*}$ is continuous and
increasing in the lattice order. Moreover,
\begin{itemize}
\item if $b_{l}^{*}\left(v\right)=v_{l-1}+1$, then $Q^{-1}\left(a_{l+1}\right)\geq F\left(b_{l}^{*}\left(v\right)|v\right)$,
\item if $b_{l}^{*}\left(v\right)<v_{l-1}+1$ and $b_{l}^{*}\left(v\right)>0$,
then $Q^{-1}\left(a_{l+1}\right)=F\left(b_{l}^{*}\left(v\right)|v\right)$,
and 
\item if $b_{l}^{*}\left(v\right)=0$ (which means that $b_{l}\left(v\right)=0$),
then $Q^{-1}\left(a_{l+1}\right)\leq F\left(0|v\right)$.
\end{itemize}
Consider a sequence $v^{0}=\left(0,0,...,0\right)$ and $v^{n}=b^{*}\left(v^{n-1}\right)$
for $n>0$. Because the sequence is bounded ($v^{n}\in V^{*}$ for
each $n$) and $b^{*}$ is continuous and increasing, it must converge
to $v^{*}=b^{*}\left(v^{*}\right)$. The properties of $b$ and $b^{*}$
functions imply that if $v_{l}^{*}>0$, then $b_{l}\left(v^{*}\right)\geq v_{l}^{*}$.
(The reason is that if $n>0$ is the first element of the sequence
such that $v_{l}^{n}>0$, then clearly $b_{l}\left(v^{n-1}\right)=v_{l}^{n}>0=v_{l}^{n-1}$,
and by monotonicity, $b_{l}\left(v^{m-1}\right)\geq v_{l}^{m}$ for
each $m$.) 

Let $l_{0}=\min\left(l=1,...,L\text{ st. }v_{l}=v_{l-1}+1\right)$,
where $l_{0}=L+1$ if the set is empty. We will show that $l_{0}=1$.
On the contrary, suppose that $l_{0}>1$. Then, for each $l<l_{0}$,
$v_{l}^{*}=b_{l}^{*}\left(v\right)<v_{l-1}^{*}+1$. The properties
of $b^{*}$ stated above imply that 
\begin{align*}
Q^{-1}\left(a_{l+1}\right) & \leq F\left(v_{l}^{*}|v^{*}\right)=a_{0}+\sum_{k\geq0}\left(1-f\left(v_{k}^{*}-v_{l}^{*}\right)\right)\left(a_{k+1}-a_{k}\right)\\
 & =a_{0}+\sum_{k=0}^{l_{0}-1}\left(1-f\left(v_{k}^{*}-v_{l}^{*}\right)\right)\left(a_{k+1}-a_{k}\right),
\end{align*}
where the last equality follows from the fact that $v_{k}^{*}-v_{l}^{*}\geq1$
for $l\leq l_{0}-1$ and $k\geq l_{0}$. Multiply both sides of the
above inequality by $\left(a_{l+1}-a_{l}\right)$ and sum across all
$l=0,...,l_{0}-1$ to obtain
\begin{align*}
\sum_{l=0}^{l_{0}-1}Q^{-1}\left(a_{l+1}\right)\left(a_{l+1}-a_{l}\right) & \geq a_{0}\left(a_{l_{0}}-a_{0}\right)+\sum_{k=0}^{l_{0}-1}\sum_{l=0}^{l_{0}-1}\left(1-f\left(v_{k}^{*}-v_{l}^{*}\right)\right)\left(a_{k+1}-a_{k}\right)\left(a_{l+1}-a_{l}\right)\\
 & =a_{0}\left(a_{l_{0}}-a_{0}\right)+\frac{1}{2}\sum_{k=0}^{l_{0}-1}\sum_{l=0}^{l_{0}-1}\left(1-f\left(v_{k}^{*}-v_{l}^{*}\right)+1-f\left(v_{l}^{*}-v_{k}^{*}\right)\right)\left(a_{k+1}-a_{k}\right)\left(a_{l+1}-a_{l}\right)\\
 & =a_{0}\left(a_{l_{0}}-a_{0}\right)+\frac{1}{2}\sum_{k=0}^{l_{0}-1}\sum_{l=0}^{l_{0}-1}\left(a_{k+1}-a_{k}\right)\left(a_{l+1}-a_{l}\right)\\
 & =a_{0}\left(a_{l_{0}}-a_{0}\right)+\frac{1}{2}\left(a_{l_{0}}-a_{0}\right)^{2}=\frac{1}{2}\left(a_{l_{0}}^{2}-a_{0}^{2}\right)=\intop_{a_{0}}^{a_{l_{0}}}xdx.
\end{align*}
(The first equality is obtained by exchanging indices $k$ and $l$.
The second one is due to $f$ being balanced.) Because the LHS of
the above inequality is equal to $\intop_{a_{0}}^{a_{l_{0}}}Q^{-1}\left(x\right)dx$,
we get a contradiction with (\ref{eq:RU wave}). The contradiction
shows that $l_{0}=1$. 

Because $l_{0}=1$, $v_{1}^{*}=1>0$, and we have $v_{l}^{*}\geq v_{1}^{*}>0$
for each $l=0,...,L-1$. The properties of the sequence $v^{n}$ imply
that $b_{l}\left(v^{*}\right)\geq v_{l}^{*}>0$, which further implies
that $Q^{-1}\left(a_{l+1}\right)\geq F\left(v_{l}^{*}|v^{*}\right)$,
and, due to the definition of $Q^{_{-1}}$, that $a_{l+1}\geq Q\left(F\left(v_{l}^{*}|v^{*}\right)\right)$
for each $l$. Moreover, for each $l$, either $v_{l+1}^{*}=b_{l+1}\left(v^{*}\right)>b_{l}\left(v^{*}\right)\geq v_{l}$,
or $v_{l+1}^{*}=v_{l}^{*}+1$. In both cases, $v_{l+1}^{*}>v_{l}^{*}$.
This establishes the existence of vector $v$ with the required properties. 
\end{proof}
The next lemma strengthens the conclusion of Lemma \ref{lem:Wave function}. 
\begin{lem}
\label{cor:Contagion}Suppose that $P\left(1\right)<1$ and $x^{*}$
is strictly $RU$-dominant. For each $\eta>0$, there exist $\delta>0$,
$a^{*}\leq x^{*}+\eta$, $L<\infty,$ and a step function $\sigma:\R\rightarrow\left[0,1\right]$
such that $\sigma\left(0\right)=a^{*}$, $\sigma\left(L\right)=1$,
and, for each $x$,
\begin{equation}
\sigma\left(x-\delta\right)\geq\delta+P\left(\delta+a^{*}+\sum_{a\in\sigma^{-1}\left(\R\right)}\left(1-f\left(\sigma^{-1}\left(a\right)-x\right)\right)\left(a-a_{-}\right)\right),\label{eq:contagion wave eq}
\end{equation}
where the summation is over the consecutive steps of the step function
$\sigma$.
\end{lem}
We refer to $\sigma$ as a $\delta$-contagion wave for $P$. 
\begin{proof}
Define $P^{\delta}\left(x\right)=P\left(x\right)+\delta$ for $\delta\in\left(0,1-P\left(1\right)\right)$
and notice that for sufficiently small $\delta_{1}>0$, for each $\delta\leq\delta_{1}$,
if $a_{0}$ is the highest maximizer 
\[
a_{0}\in\sup\arg\max_{a}\intop_{0}^{a}\left(\left(P^{\delta_{1}}\right)^{-1}\left(x\right)-x\right)dx,
\]
 then, $a_{0}\leq x^{*}+\frac{1}{2}\eta$. Each $P_{\delta_{1}}$
can be approximated by a step function $Q$ such that (a) $Q\geq P$
(hence $P^{-1}\geq Q^{-1}$), (b) each step is bounded by $a_{l}-a_{l-1}\leq\frac{1}{4}\delta_{1}$,
for $l=1,...,L$, and (c) if $a^{*}$ is the highest maximizer of
\[
a^{*}\in\sup\arg\max_{x}\intop_{0}^{a}\left(Q^{-1}\left(x\right)-x\right)dx,
\]
then $a^{*}\leq x_{0}+\frac{1}{2}\eta=x^{*}+\eta$. (We omit the details
of finding such approximations.) Find $\delta_{2}>0$ s.t. $\delta_{2}\leq\frac{1}{2}\delta_{1}$
and, for each $a>a^{*}$, we have 
\[
\intop_{a^{*}}^{a}\left(Q^{-1}\left(x\right)-x-\delta_{2}\right)dx>0.
\]
Such $\delta_{2}$ exists because $Q$ is a step function and $\lim_{x\searrow a^{*}}Q^{-1}\left(x\right)>a^{*}.$ 

Let $Q_{\delta_{2}}=Q\left(x+\delta_{2}\right)$. Then, $Q_{\delta_{2}}$
is a step function that satisfies the hypothesis of Lemma \ref{lem:Wave function}.
Let $0=v_{0}<v_{1}...<v_{L}\leq L$ be the thresholds from Lemma \ref{lem:Wave function}.
Then, for each $l>0$,
\begin{align}
a_{l-1}\geq a_{l+1}-\frac{1}{2}\delta_{1} & \geq Q_{\delta_{2}}\left(a^{*}+\sum_{k\geq0}\left(1-f\left(v_{k}-v_{l}\right)\right)\left(a_{k+1}-a_{k}\right)\right)-\frac{1}{2}\delta_{1}\nonumber \\
 & =Q\left(\delta_{2}+a^{*}+\sum_{k\geq0}\left(1-f\left(v_{k}-v_{l}\right)\right)\left(a_{k+1}-a_{k}\right)\right)-\frac{1}{2}\delta_{1}\nonumber \\
 & \geq P\left(\delta_{2}+a^{*}+\sum_{k\geq0}\left(1-f\left(v_{k}-v_{l}\right)\right)\left(a_{k+1}-a_{k}\right)\right)+\frac{1}{2}\delta_{1}\nonumber \\
 & \geq P\left(\delta_{2}+a^{*}+\sum_{k\geq0}\left(1-f\left(v_{k}-v_{l}\right)\right)\left(a_{k+1}-a_{k}\right)\right)+\delta_{2}.\label{eq: cont wave 1}
\end{align}
The first inequality follows from $a_{l+1}-a_{l-1}\leq a_{l}-a_{l-1}+a_{l+1}-a_{l}\leq\frac{1}{2}\delta_{1}$;
the equality follows from $Q_{\delta_{2}}^{-1}\left(a\right)=Q^{-1}\left(a\right)-\delta_{2}$;
the second inequality follows from $Q\geq P+\delta_{1}$; and the
last inequality follows from $\delta_{2}\leq\frac{1}{2}\delta_{1}$. 

Define
\[
\sigma\left(x\right)=\begin{cases}
a^{*} & x<0\\
a_{l} & x\in[v_{l-1},v_{l})\text{ and }l=1,...,L\\
P\left(1\right)+\delta & x\geq v_{L}.
\end{cases}
\]
Find $\delta>0$ such that $\delta\leq\delta_{2}$ and $\delta\leq v_{l+1}-v_{l}$
for each $l=0,...,L-1$. Because the right-hand side of inequality
(\ref{eq:contagion wave eq}) is increasing in $x$, we have:
\begin{itemize}
\item If $x<\delta$, then $\sigma\left(x-\delta\right)=a^{*}=a_{0}$. Hence
inequality (\ref{eq:contagion wave eq}) follows from inequality (\ref{eq: cont wave 1})
for $l=1$ and the fact that $x\leq0=v_{0}<v_{1}$ .
\item If $v_{l-1}+\delta\leq x<v_{l}+\delta$ for $l=1,...,L$, then $\sigma\left(x-\delta\right)\geq a_{l-1}$.
Hence inequality (\ref{eq:contagion wave eq}) follows from inequality
(\ref{eq: cont wave 1}) and $x\leq v_{l}$.
\item If $x\geq v_{L}+\delta$, then $\sigma\left(x-\delta\right)\geq P^{*}\left(1\right)+\delta$.
Hence inequality (\ref{eq:contagion wave eq}) is satisfied automatically. 
\end{itemize}
\end{proof}

\subsection{Lattice\label{subsec:Lattice}}

We start by describing the candidate network. For each $M\geq m$,
the $\left(M,m\right)$-lattice is a network with 
\begin{itemize}
\item $N=M^{2}$ nodes from the set $I_{M}=\left\{ 1,...,M\right\} ^{2}$.
We define a distance on $I_{M}$ by 
\[
d\left(i,j\right)=\frac{1}{m}\sqrt{\sum_{l}\left(\left(i_{l}-j_{l}\right)\mod M\right)^{2}},
\]
and a ball in this metric as $B\left(i,r\right)=\left\{ y:d\left(x,y\right)\leq r\right\} .$
The subtraction ``$\text{mod}M$'' turns the lattice into a subset
of ``Euclidean torus'' $\left[0,\frac{M}{m}\right]^{2}$, 
\item connections $g_{i,j}=1\Longleftrightarrow j\in B\left(i,1\right)$. 
\end{itemize}
In the course of the proof, we will assume that there exists values
$b$ and $B$ such that $0\ll b\ll m\ll B\ll M$ and such that $B$
is divisible by $b$ and $M$ is divisible by $B$. This divisibility
assumption simplifies the proof. The theorem remains valid without
it, but the proof requires small modifications to take care of reminder
items. We omit the details.

For each $i\in I_{M}$, and two sets $U,W\subseteq I_{M}$, let 
\begin{align}
d\left(i,W\right) & =\min_{j\in W}d\left(i,j\right)\text{ and }d\left(U,W\right)=\min_{i\in U}\min_{j\in W}d\left(i,j\right).\label{eq: distance sets}
\end{align}
For each set $W$, and each $r$, define the $r$-neighborhood of
$W$:
\[
B\left(W,r\right)=\left\{ i:d\left(i,W\right)\leq r\right\} =\bigcup_{i\in W}B\left(i,r\right).
\]

For large $m$, the neighborhoods of each agent behave in a similar
way to open balls on a Euclidean plane. This is formalized as follows.
Let $B_{\R^{2}}\left(x,r\right)$ be the ball on the plane with center
$x\in\R^{2}$ and radius $r$. Let $\left|A\right|$ be a Lebesgue
measure of a measurable set $A\subseteq\R^{2}$. Let 
\[
f_{0}\left(d,r_{1},r_{2}\right)=\frac{1}{\pi}\left|B_{\R^{2}}\left(\left(0,0\right),r_{1}\right)\cap B_{\R^{2}}\left(\left(d,0\right),r_{2}\right)\right|
\]
be the measure of the intersection of two balls, with radii $r_{1}$
and $r_{2}$ respectively, separated by distance $d$, and divided
by the measure of the unit ball $B\left(\left(0,0\right),1\right)$. 
\begin{lem}
\label{lem:geometry believes}
\begin{enumerate}
\item For each $\rho>0$, there exists $C_{\rho}<\infty$ such that if $m\geq C_{\rho}$,
then for any two agents $i,j$, for any $r_{1}\leq1\leq r_{2}$, we
have
\[
\left|\frac{\left|B\left(i,r_{1}\right)\cap B\left(j,r_{2}\right)\right|}{\left|B\left(i,1\right)\right|}-f_{0}\left(d\left(i,j\right),r_{1},r_{2}\right)\right|\leq\rho.
\]
\item Function $f_{0}$ has the following properties: 
\begin{itemize}
\item $f_{0}$ is Lipschitz over $d$ and $r_{1}\leq1\leq r_{2}$, 
\item $f_{0}$ is decreasing in $d$, and
\item $f_{0}\left(d,r_{1},r_{2}\right)=0$ if $r_{1}+r_{2}\leq d$, and
$f_{0}\left(d,r_{1},r_{2}\right)=1$ if $r_{1}=1$ and $d\leq r_{2}-r_{1}$.
\end{itemize}
\item Functions $f_{1}\left(x,r_{1};r_{2}\right)=f_{0}\left(r_{2}-x,r_{1},r_{2}\right)$
for $r_{1}\leq1$ and $x\in\R$ converge uniformly to function $\lim_{r_{2}\rightarrow\infty}f_{1}\left(x,r_{1};r_{2}\right)=f_{2}\left(x,r_{1}\right)$.
In particular, for each $\rho>0$, there exists $R_{\rho}$ such that,
if $r_{1}\leq1$ and $r_{2}\geq R_{\rho}$, then,
\[
\sup_{r_{1}\leq1,x}\left|f_{2}\left(x,r_{1}\right)-f_{1}\left(x,r_{1};r_{2}\right)\right|\leq\rho.
\]
Functions $f_{1}$ and $f_{2}$ are Lipschitz over $d$ and $r_{1}\leq1$
and increasing in $x$.
\item Let $f\left(x\right)=f_{2}\left(x,1\right)$. Function $f$ is balanced
(in the sense of the definition from Section \ref{subsec:Contagion-wave}).
\end{enumerate}
\end{lem}
\begin{proof}
The properties of $f_{0},f_{1},f_{2},$ and $f$ follow from their
geometric interpretations and the fact that the counting measure on
$I_{M}$ converges weakly to the Lebesgue measure on the torus. For
example, $f_{2}\left(x,r_{1}\right)$ is a circle segment of a radius
$r_{1}$ circle with height equal to $r_{1}+x$ for $x\in\left(-r_{1},r_{1}\right)$. 
\end{proof}

\subsection{Small cubes\label{subsec:Small-cubes}}

We divide the lattice into disjoint areas that we refer to as \emph{small
cubes}. Each cube is much smaller than the diameter of the neighborhood
of each node so that the neighborhoods of nodes in the same cube are
largely overlapping. At the same time, each small cube contains a
sufficiently large number of nodes so that the distribution of payoff
shocks within the cube can be probabilistically approximated by its
expected distribution. 

Let $G$ be a $\left(M,m\right)$-lattice. Take any $b>0$, where
we intend $b\ll m$. For each real number $x$, let $\left\lfloor x\right\rfloor $
be the largest integer no larger than $x$. For each node $i$, the
set of nodes
\begin{align*}
c^{b}\left(i\right) & =\left\{ j\in\left\{ 1,...,M\right\} ^{2}:\forall_{l}\left\lfloor i_{l}/b\right\rfloor =\left\lfloor j_{l}/b\right\rfloor \right\} 
\end{align*}
is referred to as a cube that contains $i$. Any two cubes are either
disjoint or identical. Each cube $c$ is uniquely identified by a
pair of numbers $c_{l}=\left\lfloor i_{l}/b\right\rfloor $ for each
$l=1,2$ and any $i\in c$. Due to the divisibility assumption, each
cube contains exactly $b^{2}$ elements, and there are $\left(\frac{M}{b}\right)^{2}$
small cubes on the $\left(M,m\right)$-lattice. 

Let $\mathcal{G}^{b}=\left\{ c^{b}\left(i\right):i\in G\right\} $
be the set of all cubes. We refer to the elements of $\mathcal{G}^{b}$
as \emph{small cubes}, to distinguish them from the large cubes introduced
in Section \ref{subsec:Percolation-argument}. Sometimes, we treat
$\mathcal{G}^{b}$ as a network with edges
\begin{equation}
g_{c,c^{\prime}}^{b}=1\text{ iff }\sum_{l}\left|\left(c_{l}-c_{l}^{\prime}\right)\text{mod}\frac{M}{b}\right|=1.\label{eq:small cube network}
\end{equation}
This way, each cube has four neighbors. We refer to $\left(\mathcal{G}^{b},g^{b}\right)$
as a network of cubes. 

For any $c,c^{\prime}\in\mathcal{G}^{b}$, let $d^{b}\left(c,c^{\prime}\right)$
denote the length of the shortest path between $c$ and $c^{\prime}$
in the network $\left(\mathcal{G}^{b},g^{b}\right)$. For any $S\subseteq\mathcal{G}^{b}$,
let $d^{b}\left(c,S\right)=\min_{c^{\prime}\in S}d^{b}\left(c,c^{\prime}\right)$. 

For each strategy profile $a=\left(a_{i}\right)_{i}$ and each small
cube $c\in\mathcal{G}^{b}$, define 
\begin{align*}
a\left(c\right) & =\frac{1}{\left|c\right|}\sum_{i\in c}a_{i},\\
\beta^{a}\left(i\right) & =\frac{1}{\left|B\left(i,1\right)\right|}\sum_{j\in B\left(i,1\right)}a_{j}=\frac{1}{\left|B\left(i,1\right)\right|}\sum_{j:d\left(i,j\right)\leq1}a_{j},\text{ and }\\
\beta^{a}\left(c\right)= & \frac{1}{\left|c\right|}\sum_{j\in c}\beta^{a}\left(j\right)=\frac{1}{\left|c\right|}\sum_{i\in c}\frac{1}{\left|B\left(i,1\right)\right|}\sum_{j:d\left(i,j\right)\leq1}a_{j},
\end{align*}
where $a\left(c\right)$ is the average action within the cube, $\beta^{a}\left(i\right)$
is the fraction of neighbors of $i$ who choose action 1, and $\beta\left(c\right)$
is the average fraction in cube $c$. 

\subsubsection{Average fractions}

The next result shows that if the cube is sufficiently small, individual
and average fractions are similar. 
\begin{lem}
\label{lem:Beliefs in a cube}There exists an universal constant $D<\infty$
such that, if $\frac{b}{m}\leq\rho$ and $m>C_{\rho}$, where $C_{\rho}$
and is a constant from Lemma \ref{lem:geometry believes}, then, for
each profile $a$, each small cube, and each $i,j\in c$, 
\[
\left|\beta^{a}\left(i\right)-\beta^{a}\left(c\right)\right|\leq D\rho.
\]
\end{lem}
\begin{proof}
It is enough to show there exists $D<\infty$ such that $\left|\beta^{a}\left(i\right)-\beta^{a}\left(j\right)\right|\leq D\rho$
for each $i,j\in c$. Notice that 
\begin{align*}
\left|\beta^{a}\left(i\right)-\beta^{a}\left(j\right)\right| & \leq\frac{\left|B\left(i,1\right)\backslash B\left(j,1\right)\right|}{\left|B\left(i,1\right)\right|}+\frac{\left|B\left(j,1\right)\backslash B\left(i,1\right)\right|}{\left|B\left(j,1\right)\right|}.
\end{align*}
By Lemma \ref{lem:geometry believes} and the fact that $d\left(i,j\right)\leq\sqrt{2}\rho$,
the above is no larger than
\[
\leq2\rho+2\left(1-f_{0}\left(\sqrt{2}\rho,1,1\right)\right).
\]
The claim follows from the Lipschitzness of function $f_{0}$ and
the fact that $f_{0}\left(0,1,1\right)=1$. 
\end{proof}

\subsubsection{Average best response}

For each small cube $c\in\mathcal{G}^{b}$ and realization of payoff
shocks, define the empirical cdf of best response thresholds:
\[
P_{c}\left(x|\varepsilon\right)=\frac{1}{\left|c\right|}\sum_{i\in c}\mathbf{1}\left\{ \beta\left(\varepsilon{}_{i}\right)<x\right\} .
\]
(Recall that $\beta\left(\varepsilon_{i}\right)$ is the fraction
of neighbors of individual $i$ with payoff shock $\varepsilon_{i}$
that would make her indifferent between the two actions.) For $\gamma>0$,
say that a small cube $c$ is $\gamma$-bad, if there exists $x$
such that $P_{c}\left(x|\varepsilon\right)>P\left(x\right)+\gamma$;
otherwise, the cube is $\gamma$-good. 

Next, we show that if a cube is good, then the average action can
be approximated by a best response to average beliefs. 
\begin{lem}
\label{lem:Best response in a cube}There exists a constant $D<\infty$
such that if $\frac{b}{m}\leq\rho$ and $m>C_{\rho}$, where $C_{\rho}$
is a constant from Lemma \ref{lem:geometry believes}, then, for each
equilibrium profile $a$, if small cube $c$ is $\gamma$-good, then
\[
a\left(c\right)\leq\gamma+P\left(\beta^{a}\left(c\right)+D\rho\right).
\]
\end{lem}
\begin{proof}
Notice that 
\begin{align*}
a\left(c\right) & =\frac{1}{\left|c\right|}\sum_{i\in c}a_{i}\leq\frac{1}{\left|c\right|}\sum_{i\in c}\mathbf{1}\left(\beta\left(\varepsilon_{i}\right)\leq\beta_{i}^{a}\right)\leq\frac{1}{\left|c\right|}\sum_{i\in c}\mathbf{1}\left(\beta\left(\varepsilon_{i}\right)\leq\beta^{a}\left(c\right)+D\rho\right)\\
 & =P_{c}\left(\beta^{a}\left(c\right)+D\rho|\varepsilon\right)\leq\gamma+P\left(\beta^{a}\left(c\right)+D\rho\right).
\end{align*}
The first inequality comes from the fact that if $a_{i}=1$ is a best
response, then $\beta\left(\varepsilon_{i}\right)\leq\beta_{i}^{a}$,
and the second inequality is a consequence of Lemma \ref{lem:Beliefs in a cube}.
\end{proof}

\subsubsection{Behavior dominance}

The next definition and result plays an important role in extending
the contagion wave mechanics from a one-dimensional line to a two-dimensional
lattice. 

Let $\sigma$ be an increasing step function (see Section \ref{subsec:Contagion-wave})
for the definition. Let $a=\left(a_{i}\right)$ be a strategy profile.
We say that profile $a$ is \emph{$\left(W,R,\rho\right)$-dominated}
by $\sigma$ given a set $W\subseteq\mathcal{G}^{b}$ of small cubes
and $R>0$ if for each small cube $c\in\mathcal{G}^{b}$, we have
\[
a\left(c\right)\leq\sigma\left(d\left(c,W\right)-R\right)+\rho,
\]
where distance between sets is defined in (\ref{eq: distance sets}). 
\begin{lem}
\label{lem:Payoff dominance bound}There is a constant $D<\infty$
with the following property: Fix $\rho>0$. Suppose that $\frac{b}{m}<\rho$,
$R>R_{\rho}$, and $m>C_{\rho}$, where $C_{\rho}$ and $R_{\rho}$
are the constants from Lemma \ref{lem:geometry believes}. For each
increasing step function $\sigma:\R\rightarrow\left[0,1\right]$,
and for each set of small cubes $W$, if strategy profile $a$ is
$\left(W,R,\rho\right)$-dominated by $\sigma$, then for each cube
$c$, 
\[
\beta^{a}\left(c\right)\leq a^{*}+\sum_{a\in\sigma^{-1}\left(\R\right)}\left(1-f\left(\sigma^{-1}\left(a\right)+R-d\left(c,W\right)\right)\right)\left(a-a_{-}\right)+D\rho.
\]
\end{lem}
\begin{proof}
By Lemma \ref{lem:Beliefs in a cube}, there is a constant $D_{0}$
such that for any $i\in c$, 
\begin{align*}
\beta^{a}\left(c\right) & \leq\beta^{a}\left(i\right)+D_{0}\rho=a^{*}+\frac{1}{\left|B\left(i,1\right)\right|}\sum_{j\in B\left(i,1\right)}\left(a\left(j\right)-a^{*}\right)+D_{0}\rho\\
 & \leq a^{*}+\frac{1}{\left|B\left(i,1\right)\right|}\sum_{c^{\prime}:d\left(i,c^{\prime}\right)\leq1-\sqrt{2}\rho}\left|c^{\prime}\right|\left(a\left(c^{\prime}\right)-a^{*}\right)+\frac{\left|\left\{ j:1-\sqrt{2}\rho<d\left(i,j\right)<1\right\} \right|}{\left|B\left(i,1\right)\right|}+D_{0}\rho.
\end{align*}
Lemma \ref{lem:geometry believes} implies that the third term is
bounded by 
\[
\leq1-f_{0}\left(0,1-\sqrt{2}\rho,r_{2}\right)\leq D_{1}\rho
\]
for some constant $D_{1}$ to the Lipschitzness of function $f_{0}$
and $f_{0}\left(0,1,1\right)=1$. For the second term, we have
\begin{align}
 & \frac{1}{\left|B\left(i,1\right)\right|}\sum_{c^{\prime}:d\left(i,c^{\prime}\right)\leq1-\sqrt{2}\rho}\left|c^{\prime}\right|\left(a\left(c^{\prime}\right)-a^{*}\right)\nonumber \\
\leq & \rho+\frac{1}{\left|B\left(i,1\right)\right|}\sum_{c^{\prime}:d\left(i,c^{\prime}\right)\leq1-\sqrt{2}\rho}\left|c^{\prime}\right|\left(\sigma\left(d\left(c,W\right)-R\right)-a^{*}\right)\nonumber \\
\leq & \rho+\sum_{a\in\sigma\left(\R\right)}\left(a-a_{-}\right)\frac{1}{\left|B\left(i,1\right)\right|}\sum_{c^{\prime}:\begin{array}{c}
d\left(i,c^{\prime}\right)\leq1-\sqrt{2}\rho\text{ and }d\left(c^{\prime},\bigcup W\right)\geq R+\sigma^{-1}\left(a\right)\end{array}}\left|c^{\prime}\right|\nonumber \\
\leq & \rho+\sum_{a\in\sigma\left(\R\right)}\left(a-a_{-}\right)\frac{1}{\left|B\left(i,1\right)\right|}\left|\left\{ j:d\left(i,j\right)\leq1,d\left(j,\bigcup W\right)\geq R+\sigma^{-1}\left(a\right)\right\} \right|.\label{eq:frist term}
\end{align}
(Recall that $\sigma\left(\R\right)$ is the set of steps of the step
function $\sigma$.) Let $i^{*}\in\arg\min_{j\in\bigcup W}d\left(i,j\right)$.
Then, $d\left(i,i^{*}\right)=d\left(i,\bigcup W\right)$. Applied
again, Lemma \ref{lem:geometry believes} implies that
\begin{align*}
 & \frac{1}{\left|B\left(i,1\right)\right|}\left|\left\{ j:d\left(i,j\right)\leq1,d\left(j,\bigcup W\right)\geq R+\sigma^{-1}\left(a\right)\right\} \right|\\
\leq & \frac{1}{\left|B\left(i,1\right)\right|}\left|\left\{ j:d\left(i,j\right)\leq1,d\left(j,i^{*}\right)\geq R+\sigma^{-1}\left(a\right)\right\} \right|\\
= & 1-\frac{\left|B\left(i,1\right)\cap B\left(i^{*},R+\sigma^{-1}\left(a\right)-\rho\right)\right|}{\left|B\left(i,1\right)\right|}\\
\leq & 1-f_{0}\left(d\left(i,\bigcup W\right),1,R+\sigma^{-1}\left(a\right)\right)+\rho\\
= & 1-f_{1}\left(R+\sigma^{-1}\left(a\right)-\rho-d\left(i,\bigcup W\right),1;R+\sigma^{-1}\left(a\right)\right)+\rho\\
\leq & 1-f_{2}\left(R+\sigma^{-1}\left(a\right)-\rho-d\left(i,\bigcup W\right),1\right)+\rho\\
\leq & 1-f\left(R+\sigma^{-1}\left(a\right)-d\left(i,\bigcup W\right)\right)+\left(K+1\right)\rho,
\end{align*}
where $K$ is a Lipschitz constant for $f$. Hence (\ref{eq:frist term})
is not larger than 
\begin{align*}
 & \leq\sum_{a\in\sigma\left(\R\right)}\left(a-a_{-}\right)\left(1-f\left(R+\sigma^{-1}\left(a\right)-d\left(i,\bigcup W\right)\right)\right)+D_{2}\rho
\end{align*}
for some constant $D_{2}<\infty$ that may depend on the number of
steps in the step function $\sigma$. The result follows from putting
the estimates together.
\end{proof}

\subsection{Good giant component of cubes\label{subsec:Percolation-argument}}

We will show that if the lattice is sufficiently large then, with
an arbitrarily large probability, we can find a set of small cubes
that (a) contains almost all small cubes (we say that it is \emph{giant)}
(b) it is connected in the small cube network, (c) each cube in the
set is far away from bad cubes, and (d) it contains a large set of
agents for whom action 0 is dominant. Properties (b)-(c) will allow
the contagion wave to spread across the entire set $W$, property
(a) will ensure that spreading to set $W$ means spreading almost
everywhere, and property (d) will ensure that the set contains sufficiently
many ``initial infectors''. 

Formally, say that agent $x$ is \emph{extraordinary} if action 0
is strictly dominant for such an agent. A small cube $c\in\mathcal{G}^{b}$
is \emph{extraordinary} if it only consists of extraordinary agents.
In any equilibrium, $a\left(c\right)=0$ for extraordinary cube $c$. 

Say that set $W\subseteq\mathcal{G}^{b}$ of small cubes is \emph{$\left(\gamma,R\right)$-good}
if \begin{enumerate}[label=(\alph*)] 

\item the union of all small cubes in $W$ contains at least a fraction
of $\left(1-\gamma\right)$ elements of the lattice, $\left|\bigcup W\right|\geq\left(1-\gamma\right)M^{2}$,

\item it is connected as a subset of nodes on graph $\left(\mathcal{G}^{b},g^{b}\right)$
(see the definition of a small cube network in (\ref{eq:small cube network})),

\item if $c\in\mathcal{G}^{b}$ is $\gamma$-bad, then $d\left(c,c^{\prime}\right)\geq R$
for each $c^{\prime}\in W$ (in particular, each cube in $W$ is $\gamma$-good),

\item it contains a cube $c_{0}$ such that each cube $c$ s.t. $d\left(c,c_{0}\right)\leq R$
is extraordinary. \end{enumerate} 

The goal of this subsection is to prove that large good sets of small
cubes exists with a large probability: 
\begin{lem}
\label{lem:Good giant component}For each $\gamma,\rho>0,$ and $R<\infty$,
there exists constants $m_{\gamma,\rho,R},A_{\gamma,\rho,R}>0$ such
that, if $m\geq m_{\gamma,\rho,R}$ and $M\geq\left(A_{\gamma,\rho,R}\right)^{m^{6}}$,
then there exists $b$ so that $\frac{b}{m}\leq\rho$ and, if $G$
is $\left(M,m\right)$-lattice with the associated small cube network
$\mathcal{G}^{b}$, then
\[
\Prob\left(\text{there exists }\left(\gamma,R\right)\text{-good set }W\subseteq\mathcal{G}^{b}\right)\geq1-\gamma.
\]
 
\end{lem}

\subsubsection{Large cubes}

In order to find a set $W$ that is sufficiently far from bad small
cubes, we are going to contain and separate bad small cubes in sufficiently
large sets. Let $B$ be a number that is divisible by $b$, $B=kb$,
and such that $M$ is divisible by $B$. Consider a network of cubes
$\left(\mathcal{G}^{B},g^{B}\right)$ defined in the same way as described
in Section \ref{subsec:Small-cubes}. We refer to elements of $\mathcal{G}^{B}$
as \emph{large cubes} to distinguish from the elements of $\mathcal{G}^{b}$.
Let $K=\frac{M}{B}$; then the number of large cubes is $K^{2}$. 

For each set of large cubes $U\subseteq\mathcal{G}^{B}$, and for
each $R$, define the small cube $R$-interior of $U$ as the set
of small cubes that are $R$-away from nodes that do not belong to
$U$ 
\[
W\left(U,R\right)=\left\{ c\in\mathcal{G}^{b}:d\left(c,I_{M}\backslash\left(\bigcup U\right)\right)>R\right\} .
\]
Here, $\bigcup U$ is the union of all large cubes in set $U$, and
$I_{M}\backslash\left(\bigcup U\right)$ is the set of all nodes on
$\left(M,m\right)$-lattice that do not belong to one of the large
cubes in $U$. We have the following bound on the size of set $W\left(U,R\right)$.
\begin{lem}
\label{lem:Size of R interior}Suppose that $U$ is a subset of large
cubes, $U\subseteq\mathcal{G}^{B}$. Then, \textup{
\[
\frac{1}{\left|\mathcal{G}\right|}\left|\bigcup W\left(U,R\right)\right|\geq\frac{\left|U\right|}{\left|\mathcal{G}^{B}\right|}\left(1-4\frac{1}{k}\left(\frac{Rm}{b}+1\right)\right).
\]
}.
\end{lem}
\begin{proof}
Observe that 
\[
\frac{\left|\bigcup W\left(U,R\right)\right|}{\left|\mathcal{G}\right|}=\frac{\left|\bigcup\mathcal{G}^{b}\right|}{\left|\mathcal{G}\right|}\frac{\frac{\left|\bigcup W\left(U,R\right)\right|}{\left|\bigcup\mathcal{G}^{b}\right|}}{\frac{\left|W\left(U,R\right)\right|}{\left|\mathcal{G}^{b}\right|}}\frac{\left|W\left(U,R\right)\right|}{\left|W\left(U,0\right)\right|}\frac{\left|W\left(U,0\right)\right|}{\left|\mathcal{G}^{b}\right|}.
\]
The bound is a consequence of the following observations:
\begin{itemize}
\item Because all small cubes have the same cardinality, we have $\left|\bigcup\mathcal{G}^{b}\right|=\left|\mathcal{G}\right|$
and $\frac{\left|\bigcup W\left(U,R\right)\right|}{\left|\bigcup\mathcal{G}^{b}\right|}=\frac{\left|W\left(U,R\right)\right|}{\left|\mathcal{G}^{b}\right|}$. 
\item For each regular large cube $C\in U$, $W\left(C,0\right)$ consists
of $k^{2}$ small cubes, and $W\left(C,R\right)$ consists of at least
$\left(k-2\left(\frac{Rm}{b}+1\right)\right)^{2}$ small cubes. Hence
$\frac{\left|W\left(U,R\right)\right|}{\left|W\left(U,0\right)\right|}\geq1-4\frac{1}{k}\left(\frac{Rm}{b}+1\right)$.
\item Finally, notice that $\left|W\left(U,0\right)\right|=k^{2}\left|U\right|$
and $\left|\mathcal{G}^{b}\right|=k^{2}\left|\mathcal{G}^{B}\right|$.
\end{itemize}
\end{proof}
The next result shows that if $U$ is a connected component of large
cubes, then $W\left(U,R\right)$ is a connected component of small
cubes. 
\begin{lem}
\label{lem:Connected R interior}Suppose that $R<\frac{b}{m}\left(\frac{1}{2}k-1\right)$.
If a set of large cubes $U\subseteq\mathcal{G}^{B}$ is a connected
component in the network of large cubes, then the $R$-interior set
of small cubes $W\left(U,R\right)$ is a connected component in the
network of small cubes. 
\end{lem}
\begin{proof}
For each large cube $C$, let $W\left(U,R\right)\cap\left\{ c\in\mathcal{G}^{b}:c\subseteq C\right\} $
be a part of the $R$-interior that consists of small cubes which
are contained in $C$. It is clear that $W\left(U,R\right)\cap\left\{ c\in\mathcal{G}^{b}:c\subseteq C\right\} $
is connected in the network of small cubes. If $C$ and $C^{\prime}$
are two neighboring large cubes, say $C_{1}=C_{1}^{\prime}$ and $C_{2}^{\prime}=C_{2}^{\prime}+1$,
then small cubes $c$ and $c^{\prime}$ such that $c_{1}=c_{1}^{\prime}=B\left(C_{1}-1\right)+\left(\left\lceil \frac{Rm}{b}\right\rceil +1\right)$
and $c_{2}^{\prime}=c_{2}+1=BC_{2}+1$ are neighbors and they both
belong to $W\left(U,R\right)$. Hence, set $W\left(U,R\right)=\bigcup_{C\in U}W\left(U,R\right)\cap\left\{ c\in\mathcal{G}^{b}:c\subseteq C\right\} $
is connected. 
\end{proof}

\subsubsection{Percolation theory - deterministic bounds\label{subsec:Percolation-theory-on simple lattices}}

In order to establish the existence of a giant connected component
of small cubes that are not too close to bad small cubes, we turn
to the percolation theory. The percolation theory studies properties
of graphs obtained by removal of some nodes. In this paper, we are
especially interested in the size of the largest connected component
of a so-obtained graph. 

We divide the percolation theoretic arguments into two parts: deterministic
and probabilistic. 
\begin{lem}
\label{lem:separate}For each connected $S\subseteq\mathcal{G}^{B}$
$\text{st.}$ $\left|S\right|<K$, there are connected sets $\partial S\subseteq CS\subseteq\mathcal{G}^{B}\backslash S$
such that $\left|\mathcal{G}^{B}\backslash CS\right|\leq\left|S\right|^{2}$,
\[
\left\{ c\in CS:d^{B}\left(c,S\right)\leq1\right\} \subseteq\partial S\subseteq\left\{ c\in CS:d^{B}\left(c,S\right)\leq2\right\} .
\]
\end{lem}
\begin{proof}
Because $\left|S\right|<K$ is smaller than the length and width of
the network of large cubes, set $S$ can be contained in a cube of
size $\left|S\right|^{2}$ in a way that the complement of the cube
is connected and it contains at least $\left|\mathcal{G}^{B}\right|\backslash\left|S\right|^{2}$
elements. Let $CS$ be the connected component of $\mathcal{G}^{B}\backslash S$
that contains the complement of the cube. Using Lemma 1 from \cite{bollobas_percolation_2006},
we can construct a finite path $c_{0},...,c_{k}$ of neighboring cubes
in $CS$ surrounding $S$ in an intuitive way such that, if $\partial S=\left\{ c_{0},...,c_{k}\right\} $,
then $\partial S$ satisfies the required inclusions. 
\end{proof}
\begin{lem}
\label{lem:giant component bound}Suppose that $S_{1},...,S_{J}$
is a collection of connected subsets of lattice $\mathcal{G}^{B}$
such that each $\left|S_{j}\right|<K$ and such that for any $i\neq j$,
$\min_{c\in S_{i},c^{\prime}\in S_{j}}d^{B}\left(c,c^{\prime}\right)>2$
. Then, graph $\mathcal{G}^{B}\backslash\bigcup S_{j}$ contains a
connected component of size not smaller than $\left|\mathcal{G}^{B}\right|\backslash\sum_{j}\left|S_{j}\right|^{2}$. 
\end{lem}
\begin{proof}
Suppose that $S_{1},...,S_{J}$ is a collection of connected subsets
as in the statement of the lemma. For each $j\leq J$, let $\partial S_{j}\subseteq CS_{j}$
be as in Lemma \ref{lem:separate}. Let $C=\bigcap_{i}CS_{i}$. Then,
$\left|C\right|=\left|\bigcap_{i}CS_{i}\right|\geq\left|\mathcal{G}^{B}\right|\backslash\sum_{j}\left|S_{j}\right|^{2}$. 

For each $i\neq j$, suppose that $\partial S_{i}\cap CS_{j}\neq\emptyset$.
Then, $\partial S_{i}\cap CS_{j}$ is connected. Because the distance
between $S_{i}$ and $S_{j}$ is greater than 2, $\partial S_{i}\cap S_{j}=\emptyset$.
Hence, $\partial S_{i}\subseteq CS_{j}$. It follows that, if $\partial S_{i}\cap C\neq\emptyset$,
then $\partial S_{i}\subseteq C$.

It is enough to show that $C$ is connected. Take $a,b\in C$ and
construct an arbitrary path from $a=a_{0},...,a_{n}=b$ of neighboring
cubes in network $\mathcal{G}^{B}$. Such a path may go outside set
$C$ and, if so, let $l=\min\left\{ m:a_{m}\notin C\right\} $. Suppose
that $a_{l}\notin CS_{i}$ for some $i$. Then, $a_{l-1}\in\partial S_{i}\cap C$,
and, by the above argument, $\partial S_{i}\subseteq C$. Let $k=\max\left\{ m:a_{m}\notin CS_{i}\right\} $.
Such $k$ is well-defined and $k<n$ because $a_{n}=b\in C$. Hence
$a_{k+1}\in\partial S_{i}\subseteq C$. 

Because $\partial S_{i}$ is connected, the segment of the path between
$a_{l-1}$ and $a_{k+1}$ can be replaced by a path that lies completely
within $\partial S_{j}\subseteq C$. We can repeat such a modification
for any other segment of the path that lies outside of set $C$. After
finitely many modifications, we obtain a path from $a$ to $b$ that
is entirely within $C$. It follows that $C$ is connected. 
\end{proof}
For each $S\subseteq\mathcal{G}^{B}$, say that a set $S\subseteq\mathcal{G}^{B}$
is 2-connected if, for any subset $T\subseteq S$, $\min_{c\in T,c^{\prime}\in S\backslash T}d^{B}\left(c,c^{\prime}\right)\leq2$.
In other words, a 2-connected set cannot be split into two parts that
are more than $2$ away from each other. The last result in this part
provides an upper bound on the number of distinct $2$-connected sets.
\begin{lem}
\label{lem:The-number-of-connected sets}The number of distinct $2$-connected
subsets of $\mathcal{G}^{B}$ of cardinality $r$ is no larger than
$K^{2}24^{r}$.
\end{lem}
\begin{proof}
Each $r$-element $2$-connected set $S$ can be (not necessarily
uniquely) encoded as a pair of a signature $\left(t_{0},...,t_{r-1}\right)$
such that $\sum t_{i}=r-1$ and a tuple 
\[
\left(c_{0},c_{1},...,c_{t_{0}},c_{t_{0}+1},...,c_{t_{0}+t_{1}+1},...,c_{t_{0}+...+t_{l-1}+1},...,c_{t_{0}+...+t_{l}},...,c_{r}\right),
\]
where
\begin{itemize}
\item $c_{1},...,c_{t_{0}}$ is the list of all $2$-neighbors (i..e, cubes
that have $d^{B}$ distance no larger than $2$) of $c_{0}$,
\item more generally, for each $l$, $c_{t_{0}+...+t_{l-1}+1},...,c_{t_{0}+...+t_{l}}$
is a list of all $2$-neighbors of $c_{l}$ that have not yet been
listed.
\end{itemize}
The number of different signatures is no larger than $2^{r}$. Given
signature $\left(t_{0},...,t_{r}\right)$, notice that there are at
most $K^{2}$ choices of $c_{0}$; given $c_{0}$, there are at most
$12^{t_{0}}$ choices of $c_{1},...,c_{t_{0}}$ (this is because,
for each node, there are at most 12 nodes that are at most 2-away);
etc. Thus, the number of encodings is no larger than 
\[
K^{2}\cdot12^{t_{0}}\cdot...\cdot12^{t_{r-1}}=K^{2}12^{r-1}.
\]
The result follows. 
\end{proof}

\subsubsection{Percolation theory - probabilistic arguments}

Next, we consider a standard model of percolation theory, where nodes
are removed i.i.d. with probability $p\in\left(0,1\right)$. Let $\mathcal{G}_{\left(p\right)}^{B}$
denote a random graph obtained from the lattice of large cubes $\mathcal{G}^{B}$
by removing i.i.d. nodes. The following two results provide the bounds
on the probability of the existence of a giant component of $\mathcal{G}_{\left(p\right)}^{B}$. 
\begin{lem}
\label{lem:Percolation theory}There exists a universal constant $\xi<\infty$
such that, for each $\gamma\in\left(0,1\right),K,$ and $p$, if $p\leq\xi\gamma^{2}$
and $K^{2}2^{-K}\leq\frac{1}{2}\gamma$, then 
\[
\Prob\text{\ensuremath{\left(\mathcal{G}_{\left(p\right)}^{B}\text{ has a connected component of size not smaller than }\left(1-\gamma\right)K^{2}\right)}}\geq1-\gamma.
\]
\end{lem}
\begin{proof}
Let $E\subseteq I_{K}$ be the (random) set of nodes removed to obtain
graph $\mathcal{G}_{\left(p\right)}^{B}$. For each removed node $a\in E$,
let $S\left(a\right)\subseteq E$ be the maximally 2-connected component
of removed nodes that contains $a$. In other words, $S\left(a\right)$
is $2$-connected, and if $c\in E$ is such that $d^{B}\left(c,S\left(a\right)\right)\leq2$,
then $c\in S\left(a\right)$. Let $\mathcal{S}=\left\{ S\left(a\right):a\in E\right\} $
be a collection of such components. The construction ensures that,
for each $S,T\in\mathcal{S}$, if $S\neq T$, then $\min_{c\in S,c^{\prime}\in T}d^{B}\left(c,c^{\prime}\right)>2$. 

Let $r_{\max}=\max_{S\in\mathcal{S}}\left|S\right|$. Let 
\begin{align*}
X_{r} & =\left|\left\{ S\in\mathcal{S}:\left|S\right|\geq r\right\} \right|\text{ for each }r\geq1,\\
X & =\sum_{S\in\mathcal{S}}\left|S\right|^{2}=\sum_{r}r^{2}\left(X_{r}-X_{r+1}\right)=\sum_{r}\left(r^{2}-\left(r-1\right)^{2}\right)X_{r}=\sum_{r}\left(2r-1\right)X_{r}.
\end{align*}

We compute the expected value of $X$. By Lemma \ref{lem:The-number-of-connected sets},
the number of $r$-element 2-connected sets is bounded by $K^{2}24^{r}$.
The probability that all elements of a particular $r$-element tuple
are removed is equal to $p^{r}$. The linearity of the expectation
implies that $\E X_{r}\leq K^{2}\left(24p\right)^{r}$and 
\begin{align*}
\E X & =\sum_{r}\left(2r-1\right)\E X_{r}\leq K^{2}\sum_{r}2^{r}\left(24p\right)^{r}\leq K^{2}\frac{48p}{1-48p}
\end{align*}
The probability that there exists a 2-connected component not smaller
than $K$ is not larger
\[
\Prob\left(r_{\max}\geq K\right)\leq\E X_{K}\leq K^{2}\left(24\right)^{K}.
\]

By Lemma \ref{lem:giant component bound}, the probability that $\mathcal{G}_{\left(p\right)}^{B}$
does not have a connected component not smaller than $\left(1-\gamma\right)\left|\mathcal{G}^{B}\right|$
is not larger than
\begin{align*}
\leq & \Prob\text{\ensuremath{\left(X\geq\gamma K^{2}\right)}}+\Prob\left(r_{\max}\geq K\right)\leq\frac{\E X}{\gamma K^{2}}+\Prob\left(r_{\max}\geq K\right)\leq\frac{1}{\gamma}\frac{48p}{1-48p}+K^{2}\left(24\right)^{K}.
\end{align*}
(The second inequality is due to the Markov inequality.) Hence, assuming
that $\gamma<1$, the result holds if $p\leq\frac{1}{300}\gamma^{2}$
and $K^{2}2^{-K}\leq\frac{1}{2}\gamma$.
\end{proof}
Next, we find a probability bound on the existence of a giant component
of large cubes that do not have any bad small cubes. A large cube
$C\in\mathcal{G}^{B}$ is $\gamma$-clean if it does not contain any
$\gamma$-bad small cube. Let $\mathcal{G}_{\gamma}^{B}$ be the random
subgraph of the network of large cubes that consists only of $\gamma$-clean
cubes. 
\begin{lem}
\label{lem:Percolation theory: main} There exists a universal constant
$\xi<\infty$ such that, if $b\geq\frac{1}{2\gamma}\left(\log\frac{\xi k^{2}}{\gamma^{2}}\right)^{1/2}$
and $K^{2}2^{-K}\leq\frac{1}{2}\gamma$, then 
\[
\Prob\text{\ensuremath{\left(\mathcal{G}_{\gamma}^{B}\text{ has a connected component of }\gamma\text{-clean large cubes and size at least}\left(1-\gamma\right)\left|\mathcal{G}^{B}\right|\right)}}\geq1-\gamma.
\]
\end{lem}
The giant component from the lemma is obviously uniquely defined.
We refer to it as $U_{\gamma}$.
\begin{proof}
Due to the Dvoretzky--Kiefer--Wolfowitz--Massart inequality, the
probability that a small cube $c$ is $\gamma$-bad is bounded by
\[
\Prob\left(c\text{ is }\gamma\text{-bad}\right)\leq\euler^{-2b^{2}\gamma^{2}}.
\]
 The probability that a large cube $C$ is not $\gamma$-clean is
bounded by 
\[
\Prob\left(C\text{ is not }\gamma\text{-clean}\right)\leq k^{2}\euler^{-2b^{2}\gamma^{2}}.
\]
By Lemma \ref{lem:Percolation theory} and some algebra, the claim
holds if $K^{2}2^{-K}\leq\frac{1}{2}\gamma$ and $k^{2}\euler^{-2b^{2}\gamma^{2}}\leq\frac{1}{\xi}\gamma^{2}$
for some universal constant $\xi<\infty$.
\end{proof}

\subsubsection{Extraordinary set}

A large cube $C\in\mathcal{G}^{B}$ is extraordinary if it only consists
of extraordinary agents. The next result bounds the probability that
the large component identified in the previous section contains an
extraordinary large cube.
\begin{lem}
\label{lem:Seed}There exists a universal constant $\xi<\infty$ such
that, if \textup{$\euler^{-\left(1-\gamma\right)K^{2}P\left(0\right)^{k^{2}b^{2}}}\leq\frac{1}{2}\gamma$,}
$b\geq\frac{2}{\gamma}\left(\log\frac{\xi k^{2}}{\gamma^{2}}\right)^{1/2}$,
and $K^{2}2^{-K}\leq\frac{1}{4}\gamma$, then 
\[
\Prob\text{\ensuremath{\left(\left|U_{\gamma}\right|\geq\left(1-\gamma\right)K^{2}\text{ and }U_{\gamma}\text{ contains an extraordinary large cube}\right)}}\geq1-\gamma.
\]
\end{lem}
\begin{proof}
The probability that a single agent is extraordinary is $P\left(0\right)=\Prob\left(\beta\left(\varepsilon{}_{i}\right)\leq0\right)$.
The probability that a cube $C\in\mathcal{G}^{B}$ is extraordinary
is $P\left(0\right)^{\left(kb\right)^{2}}$. Because each extraordinary
cube is also $\gamma$-clean, the probability that $C$ is extraordinary
conditionally on $C$ being part of the giant component $U_{\gamma}$
and on an arbitrary realization of payoff shocks outside of $C$ is
no smaller than $P\left(0\right)^{\left(kb\right)^{2}}$. Conditionally
on $\left|U_{\gamma}\right|\geq\left(1-\gamma\right)K^{2}$, the probability
that the giant component has no extraordinary cube is bounded by 
\begin{align*}
 & \Prob\left(U_{\gamma}\text{ has no extraordinary cube}|\left|U_{\gamma}\right|\geq\left(1-\gamma\right)K^{2}\right)\\
\leq & \left(1-P\left(0\right)^{\left(kb\right)^{2}}\right)^{\left(1-\gamma\right)K^{2}}\leq\euler^{-\left(1-\gamma\right)K^{2}P\left(0\right)^{k^{2}b^{2}}}.
\end{align*}
The claim follows from the above bound and Lemma \ref{lem:Percolation theory: main}.
\end{proof}

\subsubsection{Proof of Lemma \ref{lem:Good giant component}}

Assume w.l.o.g. that $R\geq1$ and $\gamma,\rho<1$. Let $k_{m}=\left\lceil \frac{100}{\gamma}Rm\right\rceil $
and $b_{m}=\left\lceil \frac{20}{\gamma}\left(\log\frac{100\xi k_{m}^{2}}{\gamma^{2}}\right)^{1/2}\right\rceil $,
where $\xi$ is the constant from Lemma \ref{lem:Seed}. Then, $k_{m},b_{m}\geq1$
and there is a constant $m_{\gamma,\rho,R}$ such that, if $m\geq m_{\gamma,\rho,R}$,
then $\frac{b_{m}}{m}\leq\rho$. Moreover, the assumptions of Lemma
\ref{lem:Connected R interior} are satisfied:
\begin{align*}
\frac{b_{m}}{m}\left(\frac{1}{2}k_{m}-1\right) & \geq\frac{k_{m}}{2m}-\frac{b_{m}}{m}\geq\frac{50}{\gamma}R-\rho>R.
\end{align*}

Find constant $A_{\gamma,\rho,R}<\infty$ such that for each $m\geq m_{\gamma,\rho,R}$,
\begin{align*}
\left(A_{\gamma,\rho,R}\right)^{m^{6}} & \geq k_{m}b_{m}\max\left(20,2\log2\left(-\log\left(\frac{1}{40}\gamma\right)\right),\frac{2}{1-\gamma}\left(-\log\left(\frac{1}{20}\gamma\right)\right)\left(P\left(0\right)\right)^{-k_{m}^{2}b_{m}^{2}}\right).
\end{align*}
(Such a constant exists because $k_{m}\leq\frac{200}{\gamma}Rm$ and
$b_{m}\leq m$.) Take $K\geq K_{m}=\left\lceil \frac{1}{k_{m}b_{m}}\left(A_{\gamma,\rho,R}\right)^{m^{6}}\right\rceil $
and let $M=Kk_{m}b_{m}$. Then, the assumptions of Lemma \ref{lem:Seed}
are satisfied with $\frac{1}{10}\gamma$ instead of $\gamma$:
\begin{align*}
\euler^{-\left(1-\gamma\right)K^{2}P\left(0\right)^{k_{m}^{2}b_{m}^{2}}} & \leq\frac{1}{20}\gamma\text{ and }K^{2}2^{-K}\leq2^{-\frac{1}{2}K}\leq\frac{1}{40}\gamma.
\end{align*}

Finally, 
\[
2\frac{b_{m}}{M}+4\frac{1}{k_{m}}\left(\frac{Rm}{b_{m}}+1\right)\leq2\frac{1}{k_{m}}+4\frac{Rm}{k_{m}}+\frac{4}{100}\gamma\leq\gamma,
\]
which implies that the bound in the brackets of Lemma \ref{lem:Size of R interior}
is larger than $1-4\frac{1}{k_{m}}\left(\frac{Rm}{b_{m}}+1\right)\geq1-\gamma$. 

Lemma \ref{lem:Seed} implies that 
\[
\Prob\left(\left|U_{\gamma}\right|\geq\left(1-\frac{1}{10}\gamma\right)K^{2}\text{ and }U_{\gamma}\text{ contains a extraordinary large cube}\right)\geq1-\frac{1}{10}\gamma.
\]
 If $\left|U_{\gamma}\right|\geq\left(1-\frac{1}{10}\gamma\right)K^{2}$,
Lemma \ref{lem:Size of R interior} implies that $\left|\bigcup W\left(U_{\gamma},R\right)\right|\geq\left(1-\gamma\right)M^{2}$,
and Lemma \ref{lem:Connected R interior} implies that $W\left(U_{\gamma},R\right)$
is connected in the network of small cubes. The definition of $W\left(U_{\gamma},R\right)$
implies that each small cube that is not $\gamma$-good, and hence
not contained in $U$, is at least $R$-distant from each small cube
contained in $W\left(U_{\gamma},R\right)$. Finally, because $R<\frac{b_{m}}{m}\left(\frac{1}{10}k_{m}-1\right)$,
if $C_{0}\in U_{\gamma}$ is an extraordinary large cube, then $W\left(C_{0},R\right)$
is non-empty and it contains a small cube $c_{0}\in W\left(C_{0},R\right)\subseteq W\left(U_{\gamma},R\right)$
such that for any $c$, if $d\left(c,c_{0}\right)\leq R$, then $c\in C_{0}$
and $c$ is extraordinary. Therefore set $W\left(U_{\gamma},R\right)$
is $\left(\gamma,R\right)$-good. 

\subsection{Proof of Theorem }

Fix $\eta>0$. We are going to show that, for each $\eta>0$, there
exist constants $A,m_{0}>0$ such that, if $m\geq m_{0}$ and $M\geq A^{m^{6}}$,
and $G$ is $\left(M,m\right)$-lattice, then the probability that
there is an equilibrium $a$ on the $\left(M,m\right)$-lattice such
that $\text{Av}\left(a\right)=\frac{1}{M^{2}}\sum a\geq x^{*}+\eta$
is smaller than $\eta$. The argument for the lack of equilibria with
average action below $x^{*}-\eta$ is analogous (and it follows from
exchanging the roles for binary actions 0 and 1). Combining the two
bounds (and taking maximum of respective constants $A$ and $m_{0}$)
delivers the result. 

Apply Lemma \ref{cor:Contagion} to $\frac{1}{2}\eta$ and find $\delta>0$,
$a^{*}<x+\frac{1}{2}\eta$, $L<\infty$, and a $\delta$-contagion
wave $\sigma$ for $P$. 

Let $D\geq1$ be a constant that is larger than the sum of constants
from Lemmas \ref{lem:Best response in a cube} and \ref{lem:Payoff dominance bound}.
Choose $\rho\leq\frac{1}{D}\delta$ and $\gamma\leq\min\left(\delta,\frac{1}{4}\eta\right).$
Let $R_{\rho}$ be the constant from Lemma \ref{lem:geometry believes}.
Let $R=R_{\rho}+L$. Let $m_{0}=m_{\gamma,\rho,R}$ and $A=A_{\gamma,\rho,R}$.
Choose $m\geq m_{0}$, $M\geq A^{m^{6}}$, and $b$ be as in Lemma
\ref{lem:Good giant component}. 

Let $W$ denote a $\left(\gamma,R\right)$-good set of cubes in the
network of small cubes $\mathcal{G}^{b}$ if such a set exists. Let
$c_{0}\in W$ be the cube such that for each $c$, if $d\left(c,c_{0}\right)\leq R$,
then $c$ is extraordinary. 

Let $a$ be any equilibrium on the lattice. Let $W_{d}\subseteq W$
be a maximal subset of small cubes such that the equilibrium $a$
is $\left(W_{d},\gamma,R_{\rho}\right)$-dominated by $\sigma$. If
$W$ exists, then $c_{0}\in W_{d}$ and $W_{d}$ is non-empty. (To
see why, notice that $a\left(c\right)=0\leq\sigma\left(d\left(c,c_{0}\right)-R_{\rho}\right)$
for each extraordinary cube, including all cubes $c$ st. $d\left(c,c_{0}\right)\leq R$.
Additionally, $\sigma\left(d\left(c,c_{0}\right)-R_{\rho}\right)\geq\sigma\left(L\right)=1\geq a\left(c\right)$
for each cube $c$ such that $d\left(c,c_{0}\right)>R$.) By Lemmas
\ref{lem:Best response in a cube} and \ref{lem:Payoff dominance bound},
for each $\gamma$-good small cube $c$, 
\begin{align*}
a\left(c\right) & \leq\gamma+P\left(a^{*}+\sum_{a\in\sigma\left(\R\right)}\left(1-f\left(\sigma^{-1}\left(a\right)+R_{\rho}-d\left(c,W_{d}\right)\right)\right)\left(a-a_{-}\right)+D\rho\right)\\
 & \leq\delta+P\left(a^{*}+\sum_{a\in\sigma\left(\R\right)}\left(1-f\left(\sigma^{-1}\left(a\right)+R_{\rho}-d\left(c,W_{d}\right)\right)\right)\left(a-a_{-}\right)+\delta\right).
\end{align*}
Because $\sigma$ is a $\delta$-contagion wave (see Lemma \ref{cor:Contagion}),
the above is no larger than 
\[
\leq\sigma\left(d\left(c,W_{d}\right)-R_{\rho}-\delta\right).
\]

Suppose that $W_{d}\neq W$. Because $W$ is connected, there is a
cube $c_{d}\in W\backslash W_{d}$ such that $c_{d}$ is a neighbor
of $c_{d}^{\prime}\in W_{d}$ in the network of small cubes. Then,
$d\left(c_{d},c_{d}^{\prime}\right)\leq\rho$, and, by the triangle
inequality, $d\left(c,W_{d}\cup\left\{ c_{d}\right\} \right)\geq d\left(c,W_{d}\right)-\rho$
for any cube $c$. We have:
\begin{itemize}
\item for each $\gamma$-good cube $c$, because $\rho\leq\delta$,
\[
a\left(c\right)\leq\sigma\left(d\left(c,W_{d}\right)-R-\delta\right)\leq\sigma\left(d\left(c,W_{d}\cup\left\{ c_{d}\right\} \right)-R\right).
\]
\item for each cube $c$ that is not $\gamma$-good, we have $d\left(c,W_{d}\cup\left\{ c_{d}\right\} \right)\geq R\geq R_{\rho}+L$
due to $W_{d}\cup\left\{ c_{d}\right\} \subseteq W$. But then, $a\left(c\right)\leq1=\sigma\left(L\right)=\sigma\left(d\left(c,W_{d}\cup\left\{ c_{d}\right\} \right)-R\right)$. 
\end{itemize}
It follows that equilibrium $a$ is $\left(W_{d}\cup\left\{ c_{d}\right\} ,\gamma,R_{\rho}\right)$-dominated
by $\sigma$. But this is a contradiction with the choice of $W_{d}$
as a maximal set. 

Therefore, $W_{d}=W$, $a$ is $\left(W,\gamma,R_{\rho}\right)$-dominated
by $\sigma$, and for each $c\in W$, 
\[
a\left(c\right)\leq\sigma\left(d\left(c,W\right)-R\right)+\rho=\sigma\left(-R\right)+\rho\leq a^{*}+\frac{1}{4}\eta.
\]
Hence 
\begin{align*}
\text{Av}\left(a\right) & =\frac{1}{M^{2}}\sum a_{i}=a^{*}+\frac{1}{\left|\mathcal{G}^{b}\right|}\sum_{c\in W}\left(a\left(c\right)-a^{*}\right)+\frac{\left|I_{M}\backslash\bigcup W\right|}{M^{2}}\sum_{i\notin\bigcup W}\left(a_{i}-a^{*}\right)\\
 & \leq a^{*}+\frac{1}{4}\eta+\gamma\leq x^{*}+\eta.
\end{align*}
Because the probability that $\left(\gamma,R\right)$-good set of
small cubes exists is at least $1-\gamma\geq1-\eta$, the above inequality
demonstrates our claim.

\section{Proof of Theorem \ref{thm:RU dominant always}\label{sec:Proof-of-AllDominant}}

\subsection{Proof overview}

We formally describe the best response dynamics: initial profile and
the updating process. Next, we compute capacity-type bounds on the
dynamics, i.e., calculations (\ref{eq:computations risk-dominant always})
from the main body of the paper. We show that the reminder terms are
small. We use this to show that the average payoffs at the end of
the dynamics cannot be significantly different from $x^{*}$ and conclude
the proof of the theorem. 

\subsection{Initial profile\label{subsec:Initial-profile}}

In this part of the Appendix, we define the initial profile for the
dynamics and its properties. Let $x^{*}$ be the RU-dominant outcome.
For each relation $r\in\left\{ =,<,>\right\} $, let $E_{r}=\left\{ \varepsilon_{i}:u\left(0,x^{*},\varepsilon{}_{i}\right)\text{ }r\text{ }u\left(1,x^{*},\varepsilon{}_{i}\right)\right\} $.
Then, $E_{=}$ is the set of payoff shocks that make player indifferent
if exactly fraction $x^{*}$of their neighbors plays action $1$.
Then, because $x^{*}$ is an RU-dominant outcome, $F\left(E_{<}\right)\leq x^{*}\leq F\left(E_{<}\right)+F\left(E_{=}\right)$.
If $F\left(E_{=}\right)\neq0$, define $p=\frac{F\left(E_{<}\right)+F\left(E_{=}\right)-x^{*}}{F\left(E_{=}\right)}$.
For each player $i$, let $Y_{i}$ be the binomial i.i.d. variable
equal to 1 with probability $p$ and equal to 0 otherwise. 

Define an initial strategy profile as a function of the payoff shocks:
\begin{equation}
a_{i}^{0}=\begin{cases}
BR_{i}\left(a_{-i};\varepsilon_{i}\right) & \text{if }\left|BR_{i}\left(a_{-i};\varepsilon_{i}\right)\right|=1\\
Y_{i} & \text{otherwise.}
\end{cases}\label{eq:BR process 1}
\end{equation}
For each player $i$, let $\beta_{i}^{0}=\frac{1}{g_{i}}\sum g_{ij}a_{j}^{0}$
be the fraction of neighbors of agent $i$ who play action 1 under
profile $a_{i}^{0}$. The next result derives a probabilistic bound
on the average distance of neighborhood behaviors from the RU-dominant
outcome. 
\begin{lem}
\label{lem:beta bound}For each $\eta>0$, there exists $d>0$ such
that if $d\left(g\right)\leq d$, then 
\[
\Prob\left(\sum g_{i}\left|\beta_{i}^{0}-x^{*}\right|>\eta\left(\sum g_{i}\right)\right)<\eta.
\]
\end{lem}
\begin{proof}
Variables $a_{j}^{0}$ are independent of each other and $\E a_{j}^{0}=x^{*}$.
Hence, for each $i$,
\[
\E\left(\beta_{i}^{0}-x^{*}\right)^{2}=\sum_{j}\frac{g_{ij}^{2}}{g_{i}^{2}}\E\left(a_{j}^{0}-x^{*}\right)^{2}\leq\sum_{j}d\left(g\right)\frac{g_{ij}}{g_{i}}=d\left(g\right).
\]
By the Cauchy-Schwartz inequality, we get $\E\left|\beta_{i}^{0}-x^{*}\right|\leq2\sqrt{d\left(g\right)}.$
Let $d\left(g\right)\leq d=\frac{1}{4}\eta^{4}$. Then, by the Markov's
equality, for each $\eta$,
\[
\Prob\left(\sum g_{i}\left|\beta_{i}^{0}-x^{*}\right|>\eta\left(\sum g_{i}\right)\right)\leq\frac{\E\left(\sum g_{i}\left|\beta_{i}^{0}-x^{*}\right|\right)}{\eta\left(\sum g_{i}\right)}\leq\frac{2\sqrt{d\left(g\right)}}{\eta}\leq\eta.
\]
\end{proof}

\subsection{Best response process\label{subsec:Best-response-process}}

In this subsection, we formally define best response dynamics: starting
from the initial profile $a^{0}$, agents who play 0 but have 1 as
a best response revise their actions to 1, in an arbitrary (but fixed)
order. Assume that all players are labeled with numbers $i\in\left\{ 1,...,N\right\} $.
For all $t\geq0$, and for each $i$, let 
\begin{align}
\beta_{i}^{t} & =\frac{1}{g_{i}}\sum g_{ij}a_{j}^{t},\label{eq:BR process 2}\\
p_{i}^{t} & =P\left(\beta_{i}^{t}\right),\nonumber \\
i_{t} & =\min\left\{ i:a_{i}^{t}=0\text{ and }u\left(1,\beta_{i}^{t},\varepsilon{}_{i}\right)\geq u\left(0,\beta_{i}^{t},\varepsilon{}_{i}\right)\right\} ,\nonumber \\
a_{i}^{t+1} & =\begin{cases}
1 & \text{if }i=i_{t}\\
a_{i}^{t} & \text{otherwise}.
\end{cases}\nonumber 
\end{align}
We refer to $p_{i}^{t}$ as the expected action of agent $i$ in period
$t$. Because at most one player changes actions at each step, we
have $\left|\beta_{i}^{t}-\beta_{i}^{t+1}\right|\leq d\left(g\right)$
for each $i$. The stochastic process $\left(a^{t},\beta^{t},p^{t}\right)_{t}$
depends on the realization of payoff shocks $\varepsilon$.

If the set in the third line is empty, the process stops. Because
there are finitely many players, the dynamics must stop in a finite
time. We denote the final outcome of the process as $\left(a_{i}^{U},\beta_{i}^{U},p_{i}^{U}\right)$. 

\subsection{Main step}

For each profile of expected actions $p$, define the functional
\[
\mathcal{F}\left(p\right)=\frac{1}{2}\sum_{i,j}g_{ij}\left(p_{i}-p_{j}\right)^{2}.
\]
Clearly, $\mathcal{F}\left(p^{t}\right)\geq0$ for each $t$. Also,
define function 
\[
L\left(x\right)=\intop_{x^{*}}^{x}\left(P^{-1}\left(y\right)-y\right)dy.
\]
Because $x^{*}$ is RU-dominant, it is the unique minimizer of $L\left(x\right)$.
Hence $L\left(x^{*}\right)=0$ and $L\left(x\right)>0$ for each $x\neq x^{*}$. 

The next lemma fills calculations behind formula (\ref{eq:computations risk-dominant always})
in the main body of the paper. 
\begin{lem}
For each $t$,
\begin{align}
2\sum_{i}g_{i}L\left(p_{i}^{T+1}\right)\leq & \mathcal{F}\left(p^{0}\right)+A+2\sum_{i}g_{i}\left|\beta_{i}^{0}-x^{*}\right|+2d\left(g\right)\sum g_{i},\label{eq:main bound}
\end{align}
where $A$ is defined as 
\begin{align*}
A= & \sum_{t\leq T}\sum_{i}\left(p_{i}^{t+1}-p_{i}^{t}\right)\sum_{j}g_{ij}\sum_{s=t,t+1}\left(a_{j}^{s}-p_{j}^{s}\right).
\end{align*}
\end{lem}
\begin{proof}
Observe that for each $t$, 
\begin{align*}
 & \mathcal{F}\left(p^{t+1}\right)-\mathcal{F}\left(p^{t}\right)\\
= & \sum_{i}g_{i}\left(p_{i}^{t+1}\right)^{2}-\sum_{i}g_{i}\left(p_{i}^{t}\right)^{2}-\sum_{i,j}g_{ij}\left(p_{i}^{t+1}p_{j}^{t+1}-p_{i}^{t}p_{j}^{t}\right)\\
= & \sum_{i}g_{i}\left(p_{i}^{t+1}\right)^{2}-\sum_{i}g_{i}\left(p_{i}^{t}\right)^{2}-\sum_{i,j}g_{ij}\left(\left(p_{i}^{t+1}-p_{i}^{t}\right)p_{j}^{t+1}+p_{i}^{t}\left(p_{j}^{t+1}-p_{j}^{t}\right)\right)\\
= & \sum_{i}g_{i}\left(p_{i}^{t+1}\right)^{2}-\sum_{i}g_{i}\left(p_{i}^{t}\right)^{2}-\sum_{i}\left(p_{i}^{t+1}-p_{i}^{t}\right)\sum_{j}g_{ij}\sum_{s=t,t+1}p_{j}^{s}\\
= & \sum_{i}g_{i}\left(p_{i}^{t+1}\right)^{2}-\sum_{i}g_{i}\left(p_{i}^{t}\right)^{2}-\sum_{i}g_{i}\left(p_{i}^{t+1}-p_{i}^{t}\right)\sum_{s=t,t+1}\beta_{i}^{s}+\sum_{i}\left(p_{i}^{t+1}-p_{i}^{t}\right)\sum_{j}g_{ij}\sum_{s=t,t+1}\left(a_{j}^{s}-p_{j}^{s}\right),
\end{align*}
where, in the last line, we used $g_{i}\beta_{i}^{s}=\sum_{j}g_{ij}a_{j}^{s}$.
Summing up across $t\leq T$, we obtain
\begin{align*}
 & \mathcal{F}\left(p^{T+1}\right)-\mathcal{F}\left(p^{0}\right)=\sum_{t\leq T}\left(\mathcal{F}\left(p^{t+1}\right)-\mathcal{F}\left(p^{t}\right)\right)\\
= & \sum_{i}g_{i}\left(p_{i}^{T+1}\right)^{2}-\sum_{i}g_{i}\left(p_{i}^{0}\right)^{2}-\sum_{t\leq T}\sum_{i}g_{i}\left(p_{i}^{t+1}-p_{i}^{t}\right)\sum_{s=t,t+1}\beta_{i}^{s}+A\\
= & A+\sum_{i}g_{i}\left[\left(p_{i}^{T+1}\right)^{2}-\left(p_{i}^{0}\right)^{2}-2\intop_{p_{i}^{0}}^{p_{i}^{T+1}}P^{-1}\left(y\right)dy\right]\\
 & +\sum_{t\leq T}\left[2\intop_{p_{i}^{t}}^{p_{i}^{t+1}}P^{-1}\left(y\right)dy-\left(p_{i}^{t+1}-p_{i}^{t}\right)\sum_{s=t,t+1}\beta_{i}^{s}\right].
\end{align*}

The second term of the above is equal to 
\begin{align*}
 & \sum_{i}g_{i}\left[\left(p_{i}^{T+1}\right)^{2}-\left(p_{i}^{0}\right)^{2}-2\intop_{p_{i}^{0}}^{p_{i}^{T+1}}P^{-1}\left(y\right)dy\right]\\
= & 2\sum_{i}g_{i}\left[\intop_{p_{i}^{0}}^{p_{i}^{T+1}}ydy-\intop_{p_{i}^{0}}^{p_{i}^{T+1}}P^{-1}\left(y\right)dy\right]=2\sum_{i}g_{i}\left(L\left(p_{i}^{0}\right)-L\left(p_{i}^{T+1}\right)\right).
\end{align*}
Notice that $L\left(x^{*}\right)=L\left(P\left(x^{*}\right)\right)=0$
and $L\left(P\left(\beta_{i}^{0}\right)\right)$ is Lipschitz with
constant 1. Hence the above is no larger than 
\[
\leq-2\sum_{i}g_{i}L\left(p_{i}^{T+1}\right)+\sum g_{i}\left|\beta_{i}^{0}-x^{*}\right|.
\]

Recall that $\sup_{t\leq T}\left(\beta_{i}^{t+1}-\beta_{i}^{t}\right)\leq d\left(g\right)$.
By definition of the Lebesgue integral, 
\begin{align*}
\sum_{t\leq T}\beta_{i}^{t}\lambda\left(y:\beta_{i}^{t}\leq P^{-1}\left(y\right)<\beta_{i}^{t+1}\right) & \leq\intop_{p_{i}^{0}}^{p_{i}^{T+1}}P^{-1}\left(y\right)dy\\
 & \leq\sum_{t\leq T}\left(\beta_{i}^{t}+d\left(g\right)\right)\lambda\left(y:\beta_{i}^{t}\leq P^{-1}\left(y\right)<\beta_{i}^{t+1}\right),
\end{align*}
where $\lambda$ is the Lebesgue measure on the interval $\left[0,1\right]$.
The definition of inverse function $P^{-1}$ as well as $p_{i}^{t}=P\left(\beta_{i}^{t}\right)$
for each $t$ imply that 
\[
\lambda\left(y:\beta_{i}^{t}\leq P^{-1}\left(y\right)<\beta_{i}^{t+1}\right)=p_{i}^{t+1}-p_{i}^{t}.
\]
Hence
\[
\sum_{i}g_{i}\sum_{t\leq T}\left[2\intop_{p_{i}^{t}}^{p_{i}^{t+1}}P^{-1}\left(y\right)dy-\left(p_{i}^{t+1}-p_{i}^{t}\right)\sum_{s=t,t+1}\beta_{i}^{s}\right]\leq2d\left(g\right)\sum g_{i}.
\]

The result follows from putting the estimates together and the fact
that $\mathcal{F}\left(p_{i}^{T+1}\right)\geq0$. 
\end{proof}

\subsection{Estimates}

In this section, we provide estimates of the terms on the right-hand
side of (\ref{eq:main bound}). 
\begin{lem}
\label{lem:F(p) bound}For each $\eta>0$, there exists $d_{\eta}^{F}>0$
such that, if $d\left(g\right)\leq d_{\eta}^{F}$, then 
\[
\Prob\left(\mathcal{F}\left(p^{0}\right)>\eta\left(\sum g_{i}\right)\right)<\eta.
\]
\end{lem}
\begin{proof}
Note that 
\[
\mathcal{F}\left(p^{0}\right)\leq\sum_{i}g_{i}\left(p_{i}-x^{*}\right)^{2}\leq\sum_{i}g_{i}\delta\left(\beta_{i}^{0}\right),
\]
where $\delta\left(x\right)=\left(P\left(x\right)-x^{*}\right)^{2}$.
Note that $\delta\left(x^{*}\right)=0$. Choose $\xi>0$ such that
$\delta\left(\sqrt{\xi}\right)+\sqrt{\xi}<\frac{1}{2}\eta$. Let $d<\eta$
be small enough so that Lemma \ref{lem:beta bound} holds for $\xi$.
Then, 
\[
\Prob\left(\sum_{i:\beta_{i}^{0}\geq\sqrt{\xi}}g_{i}>\sqrt{\xi}\sum g_{i}\right)\leq\xi,
\]
and, if the event in the brackets does not hold, we have 
\[
\sum_{i}g_{i}\delta\left(\beta_{i}^{0}\right)\leq\sum_{i}g_{i}\left(\delta\left(\sqrt{\xi}\right)+\sqrt{\xi}\right)\leq\eta\left(\sum_{i}g_{i}\right).
\]
\end{proof}
To gain estimates on term $A$, we need a preliminary lemma:
\begin{lem}
\label{lem:expectation bound}For each $j$ and $s$, 
\[
\E\left(a_{j}^{s}-\max\left(x^{*},p_{j}^{s}\right)|\varepsilon{}_{-j}\right)\leq0.
\]
\end{lem}
\begin{proof}
Fix player $j$. The stochastic process $\left(a^{t},\beta^{t},p^{t}\right)_{t}$
can be defined on the probability space $\Omega=E^{N}$ composed of
the realizations of the payoff shock for each individual. Consider
an auxiliary stochastic processes $\left(a^{\prime t},\beta^{\prime t},p^{\prime t}\right)_{t}$
defined on the same probability space with the same equations (\ref{eq:BR process 1})-(\ref{eq:BR process 2})
as the original process, but with setting $a_{j}^{\prime t}\equiv a_{j}^{0}$
for each $t$. Additionally, define 
\[
a_{j}^{*t+1}=1\text{ iff }u\left(1,\max\left(x^{*},\beta_{j}^{\prime t}\right),\varepsilon{}_{j}\right)\geq u\left(0,\max\left(x^{*},\beta_{j}^{\prime t}\right),\varepsilon{}_{j}\right).
\]
So defined $a_{j}^{*t}$ depends on $\varepsilon_{-j}$ only through
process $\beta^{\prime}$. Hence, for each $\varepsilon_{-j}$,
\begin{align*}
\Prob\left(a_{j}^{*t+1}=1|\varepsilon{}_{-j}\right) & =\Prob\left(u\left(1,\max\left(x^{*},\beta_{j}^{\prime t}\right),\varepsilon_{j}\right)\geq u\left(0,\max\left(x^{*},\beta_{j}^{\prime t}\right),\varepsilon_{j}\right)|\varepsilon{}_{-j}\right)\\
 & =P\left(\max\left(x^{*},\beta_{j}^{\prime t}\right)\right).
\end{align*}

Notice that $a_{j}^{*t}\geq a_{j}^{t}$ for each $t$. Indeed, let
$t_{0}=\inf\left\{ t:a_{j}^{t}=1\right\} $ and equal $\infty$ if
the set is empty. Then, $\beta_{i}^{t}=\beta_{i}^{\prime t}$ for
each $i$ and $t<t_{0}$ . Moreover, $a_{j}^{t_{0}}=1$ implies $u\left(1,\beta_{j}^{t_{0}-1},\varepsilon{}_{j}\right)\geq u\left(0,\beta_{j}^{t_{0}-1},\varepsilon{}_{j}\right)$
, which implies that $a_{j}^{*t_{0}}=1$. 

Further, payoff complementarities imply that, for each $s$, $\beta^{\prime s}\leq\beta^{s}$,
and hence $p^{\prime s}\leq p^{s}$. Additionally, $p^{\prime s-1}\leq p^{\prime s}$.
Thus,
\begin{align*}
\E\left(a_{j}^{s}-\max\left(x^{*},p_{j}^{s}\right)|\varepsilon_{-j}\right) & =\Prob\left(a_{j}^{s}=1|\varepsilon{}_{-j}\right)-\max\left(x^{*},p_{j}^{\prime s}\right)\\
 & \leq\Prob\left(a_{j}^{*s}=1|\varepsilon{}_{-j}\right)-\max\left(x^{*},p_{j}^{\prime s}\right)\\
 & =P\left(\max\left(x^{*},\beta_{j}^{\prime s-1}\right)\right)-\max\left(x^{*},p_{j}^{\prime s}\right)\\
 & =\max\left(x^{*},p_{j}^{\prime s-1}\right)-\max\left(x^{*},p_{j}^{\prime s}\right)\leq0,
\end{align*}
where the first equality is due to the fact that $p_{j}^{\prime s-1}$
and $\beta_{j}^{\prime s-1}$ are measurable wrt. $\varepsilon{}_{-i}$.
\end{proof}
\begin{lem}
\label{lem:Probability bound 1}For each $\eta>0$, there exists $d_{\eta}>0$
such that, if $d\left(g\right)\leq d_{\eta}^{1}$, then 
\[
\Prob\left(\frac{1}{g_{i}}\sum g_{ij}\left(a_{j}^{s}-\max\left(x^{*},p_{j}^{s}\right)\right)\geq\eta\right)\leq\eta.
\]
\end{lem}
\begin{proof}
By Lemma \ref{lem:expectation bound}, finite stochastic process $X_{j}=\frac{1}{g_{i}}\sum_{j^{\prime}\leq j}g_{ij^{\prime}}a_{j^{\prime}}^{s}$
is a supermartingale. Take $d_{\eta}=-\frac{\eta}{\ln\eta}$. Then,
the Azuma-Hoeffding's Inequality implies that 
\[
\Prob\left(\frac{1}{g_{i}}\sum g_{ij}a_{j}^{s}-p_{j}^{s}\geq\eta\right)\leq\exp\left(-\frac{\eta}{\sum\frac{g_{ij}^{2}}{g_{i}^{2}}}\right)\leq\exp\left(-\frac{1}{d\left(g\right)}\eta\right)\leq\exp\left(\text{ln}\eta\right)=\eta.
\]
\end{proof}
\begin{lem}
\label{lem: A bound}For each $\eta>0$, there exists $d_{\eta}^{A}>0$
such that if $d\left(g\right)\leq d_{\eta}^{A}$, then for each $i$
and $s$, 
\[
\Prob\left(A\geq\eta\sum_{i}g_{i}\right)\leq\eta.
\]
\end{lem}
\begin{proof}
Because $p_{i}^{t+1}>p_{i}^{t}$ for each $i$, 
\begin{align*}
A= & \sum_{t\leq T}\sum_{i}\left(p_{i}^{t+1}-p_{i}^{t}\right)\sum_{j}g_{ij}\sum_{s=t,t+1}\left(a_{j}^{s}-\max\left(x^{*},p_{j}^{s}\right)\right)\\
 & +\sum_{t\leq T}\sum_{i}\left(p_{i}^{t+1}-p_{i}^{t}\right)\sum_{j}g_{ij}\sum_{s=t,t+1}\left(\max\left(x^{*},p_{j}^{s}\right)-p_{j}^{s}\right)\\
\leq & \sum_{t\leq T}\sum_{i}\left(p_{i}^{t+1}-p_{i}^{t}\right)\sum_{j}g_{ij}\sum_{s=t,t+1}\left(a_{j}^{s}-\max\left(x^{*},p_{j}^{s}\right)\right)+2\sum_{j}g_{j}\left|p_{j}^{0}-x^{*}\right|\\
= & A_{1}+A_{2}.
\end{align*}
We are going to bound each of the two terms separately. 

Let $d_{\eta}^{2}$ be the constant from Lemma \ref{lem: A bound}.
Then, if $d\left(g\right)\leq d_{\eta}^{A1}=d_{\frac{1}{8}\sqrt{\eta}}^{2}$,
\[
\E\left(\sum_{t\leq T}\sum_{i}\left(p_{i}^{t+1}-p_{i}^{t}\right)\sum_{j}g_{ij}\sum_{s=t,t+1}\left(a_{j}^{s}-\max\left(x^{*},p_{j}^{s}\right)\right)\right)\leq\frac{1}{4}\sqrt{\eta}\sum_{i}g_{i}.
\]
By Markov's inequality, 
\[
\Prob\left(A_{1}\geq\frac{1}{2}\eta\sum g_{i}\right)\leq\frac{\frac{1}{4}\sqrt{\eta}\sum_{i}g_{i}}{\frac{1}{2}\eta\sum g_{i}}\leq\frac{1}{2}\eta.
\]

Take $\delta\left(x\right)=\left|P\left(x\right)-x^{*}\right|$. Note
that $\delta\left(x^{*}\right)=0$. Choose $\xi>0$ such that $\max\left(\xi,4\left(\delta\left(\sqrt{\xi}\right)+\sqrt{\xi}\right)\right)<\frac{1}{2}\eta$.
Let $d_{\eta}^{A2}<\eta$ be sufficiently small so that Lemma \ref{lem:beta bound}
holds for $\xi$. Then, 
\[
\Prob\left(\sum_{i:\beta_{i}^{0}\geq\sqrt{\xi}}g_{i}>\sqrt{\xi}\sum g_{i}\right)\leq\xi,
\]
and, if the event in the brackets does not hold, we have 
\[
2\sum_{j}g_{j}\left|P\left(\beta_{j}^{0}\right)-x^{*}\right|\leq2\left(\delta\left(\sqrt{\xi}\right)+\sqrt{\xi}\right)\sum_{i}g_{i}\leq\frac{1}{2}\eta\left(\sum_{i}g_{i}\right).
\]

Take $d_{\eta}^{A}=\min\left(d_{\eta}^{A1},d_{\eta}^{A2}\right)$.
Then, 
\[
\Prob\left(A\geq\eta\sum_{i}g_{i}\right)\leq\Prob\left(A_{1}\geq\frac{1}{2}\eta\sum_{i}g_{i}\right)+\Prob\left(A_{2}\geq\frac{1}{2}\eta\sum_{i}g_{i}\right)\geq\eta.
\]
\end{proof}

\subsection{Average payoffs at the end of dynamics}

We show that the average payoffs when the upper best response dynamics
stop are not much higher than $x^{*}$. 
\begin{lem}
\label{lem:upper equilbrium bound}For each $\eta>0$, there exists
$d_{\eta}^{U}>0$ such that, if $d\left(g\right)\leq d_{\eta}^{U}$,
then
\[
\Prob\left(\text{Av}\left(a^{U}\right)\geq\left(\eta+x^{*}\right)\sum_{i}g_{i}\right)\leq\eta.
\]
\end{lem}
\begin{proof}
By definition, $x^{*}$ is the unique maximizer of $L\left(x\right)$.
Fix $\eta>0$ and find $\xi>0$ such that $\sqrt{\xi}\leq\eta$ and
if $L\left(x\right)\leq\sqrt{\xi}$, then $x\leq x^{*}+\frac{1}{2}\eta$.

Let$\left(a^{t},\beta^{t},p^{t}\right)_{t}$ be the upper best response
dynamics defined in Section \ref{subsec:Best-response-process}. By
Lemmas \ref{eq:main bound}, \ref{lem:F(p) bound}, and \ref{lem: A bound},
if $d\leq d_{\xi}^{U}=\max\left(d_{\xi}^{F},d_{\xi}^{A}\right)$,
then
\[
\sum_{i}g_{i}L\left(p_{i}^{U}\right)\leq\xi\sum_{i}g_{i}
\]
with a probability of at least $1-\xi$. It follows that $\sum_{i:L\left(p_{i}^{U}\right)\geq\sqrt{\xi}}g_{i}\leq\sqrt{\xi},$
which implies that $\sum_{i:\beta_{i}^{U}\geq x^{*}+\frac{1}{2}\eta}g_{i}\leq\sqrt{\xi}$.
Hence,
\[
\sum g_{i}\beta_{i}^{U}\leq\sum_{i:\beta_{i}^{U}\leq x^{*}+\frac{1}{2}\eta}g_{i}\left(x^{*}+\frac{1}{2}\eta\right)+\sqrt{\xi}\sum g_{i}\leq\left(x^{*}+\eta\right)\sum_{i}g_{i}.
\]
Finally, notice that 
\[
\text{Av}\left(a^{U}\right)=\sum_{i}g_{i}a_{i}^{U}=\sum_{i}\sum_{j}g_{ij}a_{i}^{U}=\sum_{i}\sum_{j}g_{ij}a_{i}^{U}=\sum_{i}g_{i}\beta_{i}^{U}.
\]
The result follows from the above inequality.
\end{proof}

\subsection{Proof of Theorem \ref{thm:RU dominant always}}

Lemma \ref{lem:upper equilbrium bound} shows that the best response
dynamics, where players only revise their actions upwards, stop with
a profile $a^{U}$ with average payoffs close to $x^{*}$. An analoguous
result shows that a lower version of the best response dynamics, initiated
from the same profile $a^{0}$ and where players only revise their
actions downwards, stop with a profile $a^{L}$ with average payoffs
also close to $x^{*}.$ 

Due to payoff complementarities, the lower best response dynamics
initiated from profile $a^{U}$ will stop at equilibrium profile $a^{UL}$
that lies in between $a^{U}$ and $a^{L}$. The latter implies that
the average payoffs must lie in between the average payoffs $\text{Av}\left(a^{U}\right)$
and $\text{Av}\left(a^{L}\right)$. The claim follows. 

\subsection{Extension to unweighted average\label{sec:Extension-to-unweighted}}

The argument remains identical except for the following modification
of Lemma \ref{lem:upper equilbrium bound}: For each $\eta>0$ and
$w<\infty$, there exists $d_{\eta}^{U}>0$ such that, if $d\left(g\right)\leq d_{\eta}^{U}$,
and $w\left(g\right)\leq w$ then
\[
\Prob\left(\text{Av}_{\text{unweighted}}\left(Ua^{0}\right)\geq\left(\eta+x^{*}\right)\right)\leq\eta.
\]

To see the above claim, recall that $a_{i}^{U}\geq a_{i}^{0}$ . Hence
\begin{align*}
 & \text{Av}_{\text{unweighted}}\left(a^{U}\right)-\text{Av}_{\text{unweighted}}\left(a\right)\\
= & \frac{1}{N}\sum_{i}\left(a_{i}^{U}-a_{i}^{0}\right)=\frac{1}{\min_{i}g_{i}}\frac{1}{N}\sum_{i}\left(\min_{j}g_{j}\right)\left(a_{i}^{U}-a_{i}^{0}\right)\\
\leq & \frac{1}{\min_{i}g_{i}}\frac{1}{N}\sum_{i}g_{i}\left(a_{i}^{U}-a_{i}^{0}\right)\leq\frac{1}{\min_{i}g_{i}}\frac{\sum g_{i}}{N}\frac{1}{\sum g_{i}}\sum_{i}g_{i}\left(a_{i}^{U}-a_{i}^{0}\right)\\
\leq & \frac{\max_{i}g_{i}}{\min_{i}g_{i}}\left(\text{Av}\left(Ua\right)-\text{Av}\left(a^{0}\right)\right)=w\left(g\right)\left(\text{Av}\left(Ua\right)-\text{Av}\left(a^{0}\right)\right).
\end{align*}
An application of Lemma \ref{lem:upper equilbrium bound} established
the claim. 

\bibliographystyle{ecca}
\bibliography{\string"C:/Users/Marcin/Documents/A My work/A pisanina/bibliography/RUNetworks\string"}

\begin{thebibliography}{12}
\providecommand{\natexlab}[1]{#1}

\bibitem[{Blume(1993)}]{blume_statistical_1993}
\textsc{Blume, L.~E.} (1993). The statistical mechanics of strategic
  interaction. \textit{Games and economic behavior}, \textbf{5}~(3), 387--424,
  publisher: Elsevier.

\bibitem[{Bollobás \textit{et~al.}(2006)Bollobás, Riordan and
  Riordan}]{bollobas_percolation_2006}
\textsc{Bollobás, B.}, \textsc{Riordan, O.} and \textsc{Riordan, O.} (2006).
  \textit{Percolation}. Cambridge University Press.

\bibitem[{Ellison(1993)}]{ellison_learning_1993}
\textsc{Ellison, G.} (1993). Learning, local interaction, and coordination.
  \textit{Econometrica: Journal of the Econometric Society}, pp. 1047--1071,
  publisher: JSTOR.

\bibitem[{Ellison(2000)}]{ellison_basins_2000}
\textsc{---} (2000). Basins of attraction, long-run stochastic stability, and
  the speed of step-by-step evolution. \textit{The Review of Economic Studies},
  \textbf{67}~(1), 17--45, publisher: Wiley-Blackwell.

\bibitem[{Harsanyi and Selten(1988)}]{harsanyi_general_1988}
\textsc{Harsanyi, J.~C.} and \textsc{Selten, R.} (1988). A general theory of
  equilibrium selection in games. \textit{MIT Press Books}, \textbf{1},
  publisher: The MIT Press.

\bibitem[{Jackson and Zenou(2015)}]{jackson_games_2015}
\textsc{Jackson, M.~O.} and \textsc{Zenou, Y.} (2015). Games on networks. In
  \textit{Handbook of game theory with economic applications}, vol.~4,
  Elsevier, pp. 95--163.

\bibitem[{Kandori \textit{et~al.}(1993)Kandori, Mailath and
  Rob}]{kandori_learning_1993}
\textsc{Kandori, M.}, \textsc{Mailath, G.~J.} and \textsc{Rob, R.} (1993).
  Learning, mutation, and long run equilibria in games. \textit{Econometrica:
  Journal of the Econometric Society}, pp. 29--56, publisher: JSTOR.

\bibitem[{Morris(2000)}]{morris_contagion_2000}
\textsc{Morris, S.} (2000). Contagion. \textit{The Review of Economic Studies},
  \textbf{67}~(1), 57--78, publisher: Wiley-Blackwell.

\bibitem[{Newton(2021)}]{newton_conventions_2021}
\textsc{Newton, J.} (2021). Conventions under heterogeneous behavioural rules.
  \textit{The Review of Economic Studies}, \textbf{88}~(4), 2094--2118.

\bibitem[{Peski(2010)}]{peski_generalized_2010}
\textsc{Peski, M.} (2010). Generalized risk-dominance and asymmetric dynamics.
  \textit{Journal of Economic Theory}, \textbf{145}~(1), 216--248, publisher:
  Elsevier.

\bibitem[{Vershynin(2018)}]{vershynin_high-dimensional_2018}
\textsc{Vershynin, R.} (2018). \textit{High-{Dimensional} {Probability}: {An}
  {Introduction} with {Applications} in {Data} {Science}}. Cambridge {Series}
  in {Statistical} and {Probabilistic} {Mathematics}, Cambridge: Cambridge
  University Press.

\bibitem[{Young(1993)}]{young_evolution_1993}
\textsc{Young, H.~P.} (1993). The evolution of conventions.
  \textit{Econometrica: Journal of the Econometric Society}, pp. 57--84,
  publisher: JSTOR.

\end{thebibliography}

\end{document}